\documentclass[a4paper,twocolumn,11pt,accepted=2026-04-03]{quantumarticle}
\pdfoutput=1

\usepackage{graphicx}
\usepackage{amsmath}
\usepackage{amssymb}
\usepackage{amsthm}
\usepackage{mathtools}
\usepackage{bm}
\usepackage{bbold}
\usepackage{braket}
\usepackage{tcolorbox}
\usepackage{tabularray}

\usepackage[colorlinks]{hyperref}
\hypersetup{
plainpages=true,
breaklinks=true,
hypertexnames=false,
pageanchor=true,
colorlinks=true,
linkcolor={blue},
citecolor={red},
urlcolor={blue},
anchorcolor={black}
}
\usepackage{cleveref}

\newtheorem{definition}{Definition}[section]
\newtheorem{lemma}{Lemma}[section]
\newtheorem{proposition}{Proposition}[section]
\newtheorem{theorem}{Theorem}[section]

\newcommand{\tr}{\operatorname{tr}}
\newcommand{\gr}{\operatorname{Gr}}

\newcommand{\sgn}{\operatorname{sgn}}

\newcommand{\partialp}{\partial_{\theta_p}}
\newcommand{\partialps}{\partial_{\theta_p^*}}
\newcommand{\partialpp}{\partial_{\theta_{p'}}}

\newcommand{\argmax}[1]{\underset{#1}{\operatorname{arg\,max }}}

\begin{document}

\author{Adrien Kahn}
\affiliation{CPHT, CNRS, Ecole Polytechnique, Institut Polytechnique de Paris, 91120 Palaiseau, France}
\affiliation{Collège de France, Université PSL, 11 place Marcelin Berthelot, 75005 Paris, France}
\affiliation{Inria Paris-Saclay, Bâtiment Alan Turing, 1, rue Honoré d’Estienne d’Orves – 91120 Palaiseau}
\affiliation{LIX, CNRS, École polytechnique, Institut Polytechnique de Paris, 91120 Palaiseau, France}
\author{Luca Gravina}
\affiliation{Institute of Physics, Ecole Polytechnique Fédérale de Lausanne (EPFL), CH-1015 Lausanne, Switzerland}
\affiliation{Center for Quantum Science and Engineering, Ecole Polytechnique Fédérale de Lausanne (EPFL), CH-1015 Lausanne, Switzerland}
\author{Filippo Vicentini}
\affiliation{CPHT, CNRS, Ecole Polytechnique, Institut Polytechnique de Paris, 91120 Palaiseau, France}
\affiliation{Collège de France, Université PSL, 11 place Marcelin Berthelot, 75005 Paris, France}
\affiliation{Inria Paris-Saclay, Bâtiment Alan Turing, 1, rue Honoré d’Estienne d’Orves – 91120 Palaiseau}
\affiliation{LIX, CNRS, École polytechnique, Institut Polytechnique de Paris, 91120 Palaiseau, France}

\title{Variational subspace methods and application to improving variational Monte Carlo dynamics}

\maketitle

\onecolumngrid

\begin{abstract}
We present a formalism that allows for the direct manipulation and optimization of subspaces, circumventing the need to optimize individual states when using subspace methods. Using the determinant state mapping, we can naturally extend notions such as distance and energy to subspaces, as well as Monte Carlo estimators, recovering the excited states estimation method proposed by Pfau et al. As a practical application, we then introduce Bridge, a method that improves the performance of variational dynamics by extracting linear combinations of variational time-evolved states. We find that Bridge is both computationally inexpensive and capable of significantly mitigating the errors that arise from discretizing the dynamics, and can thus be systematically used as a post-processing tool for variational dynamics.
\end{abstract}

\vspace{0.5cm}
\twocolumngrid

\section{Introduction}

The simulation of quantum many-body systems is a major challenge of modern computational physics~\cite{Bookatz2014QMA}. 
Although brute-force diagonalization is incapable of tackling Hilbert spaces with more than a few particles, several families of techniques have been proposed over the years to provide approximate solutions. 
In particular, variational methods, that rely on the optimization of a parametrized wavefunction, have appeared as efficient tools in approximating solutions of quantum many-body problems, such as dynamics or ground state estimation. 
Families of variational states that have been studied in recent years include tensor networks~\cite{Bridgeman2017Handwaving, Ors2014PracticalTensorNetworks, Schollwock2011}, variational quantum circuits~\cite{Tilly2022VQEReview, Peruzzo2014VQE, McClean2016VQE}, and neural quantum states~\cite{Carleo2017, ModernApplications}. 
While tensor networks are well suited for low-dimensional problems with limited entanglement, such as ground state estimation for 1D gapped Hamiltonians~\cite{Hastings20071dAreaLaw}, neural network quantum states have recently appeared as a powerful alternative. 
They have been successful in approximating the ground state of various spin~\cite{Choo2019J1J2, Viteritti2023ViT, Chen2024MinSR}, bosonic~\cite{Denis2024Bosons} and fermionic~\cite{Luo2019Backflow, RobledoMoreno2022HiddenFermion} systems.

In the context of dynamics, generally considered a more challenging problem for variational methods, neural quantum states have also proven useful. 
Several variational schemes have been proposed to simulate dynamics with neural quantum states, such as the time-dependent variational principle~\cite{Carleo2017Boson,Schmitt2020MarkusMarkus,Schmitt2023}, the projected time-dependent variational principle~\cite{Sinibaldi2023Unbiasing, Gravina2024PTVMC}, or global-in-time variational principles~\cite{Walle2024Global, Sinibaldi2024Galerkin}.

Whether it be exact diagonalization or variational Monte Carlo, many methods that operate on quantum states can be improved by subspace methods.
One of the most prominent examples of subspace method is the Lanczos algorithm that improves upon the power method by considering linear combinations of the generated vectors~\cite{Golub, Saad2011, TrefethenBau, Minganti2022FasterThanTheClockLanczos}. 
Lanczos or Lanczos-inspired methods have been successfully used in conjunction with variational Monte Carlo, to refine ground state estimations~\cite{Wang2025Lanczos, Chen2022Lanczos, Sorella2001LanczosNonsense} and to simulate dynamics~\cite{Paeckel2019MPSDynamics} among other applications~\cite{Melo2025VPTLanczos}.
Subspace methods in variational Monte Carlo typically represent the relevant subspace as the span of a given basis of states, each having been obtained by an individual optimization.
For constructing the Krylov subspace~\cite{Wang2025Lanczos} or for excited states optimization~\cite{Choo2018SymmetrizeOtherWay} for example, the sequential nature of optimizations causes an accumulation of errors for later basis states.
This problem stems from the lack of a formal language to operate directly on subspaces, instead of individual states.

In this paper, we formalise a framework where the fundamental object is the subspace itself, thus providing for a more natural language to derive variational Monte Carlo algorithms operating directly on subspaces.
Inspired by a recently proposed method to optimize for several excited states at once \cite{Pfau2024}, we present the \textit{determinant state mapping}, that allows the direct manipulation of subspaces by encoding them into a single quantum state of a larger Hilbert space. 
We find that this framework not only shines a new light on some previously known results, but also establishes new methods as natural generalization of single quantum state methods.

We then leverage this framework to introduce \textit{Bridge}, a post-processing method for variational dynamics.
Using previously generated variational states approximating quantum dynamics, we improve their performance by taking linear combinations of these states \cite{Sinibaldi2024Galerkin}.
This method is inexpensive compared to the cost of generating the original states, such that it can be used systematically at a negligible cost. 
We show that in certain circumstances, Bridge can improve the performance of the original variational states by several orders of magnitude. 
Crucially, we find that this method relies on a specific estimator derived from determinant state framework in order to obtain consistent performances as the number of states increases.

The manuscript is structured as follows.
In \cref{sec:determinant_state} we introduce the formalism of the determinant state and several of its applications. 
In \cref{sec:bridge}, we present Bridge, and we demonstrate its performance on the square lattice transverse-field Ising model in \cref{sec:numerical_results}. \Cref{sec:determinant_state} and the later sections are relatively independent.

\section{Subspace methods with the determinant state}
\label{sec:determinant_state}

In quantum physics, a state is represented by a vector of the Hilbert space $\mathcal H$, up to its norm and global phase, that are irrelevant. 
In other words, quantum states are effectively one-dimensional subspaces of $\mathcal H$. 
In this section, we examine general $m$-dimensional subspaces of $\mathcal H$ and generalise the formalism of pure state quantum mechanics to them.   
This framework goes beyond purely mathematical considerations, and extends to numerical methods.
We will show that it recontextualizes the excited states variational Monte Carlo method proposed by Pfau et al. \cite{Pfau2024} as the natural generalization of variational Monte Carlo to $m$-dimensional subspaces, but also allows for the formulation of new methods, as we will illustrate by proposing a method to construct a basis for dynamics.

\subsection{Encoding subspaces as determinant states}
\label{sec:det_state_mapping}

A natural place to start in order to represent a subspace $V$ is with a basis $\left\{ \ket{\phi_k} \right\}_k$ of $V$, such that $V = \operatorname{span} \left( \left\{ \ket{\phi_k} \right\}_k \right)$.  
The exponential cost of representing each basis vector can be managed by means of variational encodings (tensor networks, neural networks...) of each basis state $\ket{\phi_k}$.

Regardless of how the basis vectors are encoded, this representation of a subspace relies on an arbitrary choice of basis, such that the mapping
\begin{equation}
V \mapsto \left\{ \ket{\phi_k} \right\}_k
\end{equation}
is not well-defined.
For example, any permutation of the basis states $\ket{\phi_k}$ encodes the same subspace.
More generally, given any complex $m \times m$ invertible matrix $B$ that represents a change of basis, the family $\left\{ \ket{\varphi_k} \right\}_k$ defined by
\begin{equation}
\label{eq:change-of-basis}
\ket{\varphi_p} = \sum_k B_{kp} \ket{\phi_k}
\end{equation}
would parametrize the same subspace, since $\operatorname{span} \left( \left\{ \ket{\phi_k} \right\}_k \right) = \operatorname{span} \left( \left\{ \ket{\varphi_k} \right\}_k \right)$.
In order to have a well-defined and bijective encoding, we should instead represent a subspace by the set of all its bases.
Formally, we define on $\mathcal H^m$, the set of families of $m$ vectors of $\mathcal H$, the equivalence relation $\sim$ such that two families are equivalent if and only if they span the same subspace. Then, the following map is a well-defined bijective encoding of a subspace:
\begin{equation}
\begin{cases}
\gr_m(\mathcal H) &\to \mathcal H^m / \sim \\
V &\mapsto \left[ \left\{ \ket{\phi_k} \right\}_k \right]
\end{cases},
\end{equation}
where $\gr_m(\mathcal H)$ is the set of subspaces of $\mathcal H$ of dimension $m$, known as the Grassmannian, and $\left[ \left\{ \ket{\phi_k} \right\}_k \right]$ is the equivalence class of $\left\{ \ket{\phi_k} \right\}_k$.

In practice, however, we can only manipulate a representative of a class, and not the class itself.
This implies that that any expression involving a subspace should be independent from the choice of representative.

However, working with representatives in a representative-independent way is well-developed in the context of quantum states.
A quantum state is an element of the projective space $\gr_1 (\mathcal H)$, that is, it is equivalence classes under multiplication by a scalar. 
We will now show that it is possible to transform equivalence classes of $\mathcal H^m$ under change of basis into equivalence classes of $\mathcal H^{\otimes m}$ under multiplication by a scalar. 
This will effectively allow us to map a subspace onto a quantum state.

Inspired by Pfau et al. \cite{Pfau2024}, we propose to encode the subspace $V = \operatorname{span} \left( \left\{ \ket{\phi_k} \right\}_k \right)$ with the quantum state $\bm{\ket{\phi_A}}$ given by the antisymmetrization of a given basis,
\begin{equation}
\bm{\ket{\phi_A}} = \bm{S_-} \ket{\phi_1} \otimes \ldots \otimes \ket{\phi_m},
\end{equation}
where $\bm{S_-}$ is the antisymmetrization operator (see \cref{def:antisymmetrizer}).
The central motivation for this encoding is that given two different bases $\left\{ \ket{\phi_k} \right\}_k$ and $\left\{ \ket{\varphi_k} \right\}_k$ of the same subspace, related by a change of basis matrix $B$ (see \cref{eq:change-of-basis}), the antisymmetrization of the two bases will differ only by a complex scalar factor $\det B$, that is,
\begin{equation}
\label{eq:det_cob}
\bm{\ket{\varphi_A}} = \det B \bm{\ket{\phi_A}}.
\end{equation}
The proof is reported in \cref{prop:slater_cob}.
The two corresponding quantum states of $\mathcal H^{\otimes m}$ are therefore identical.

For example, it is easy to see that if the basis states $\ket{\phi_k}$ are permuted, their antisymmetrization will be modified by a factor $\pm 1$.
The equivalence classes of $\mathcal H^m$ under change of basis have thus been turned into equivalence classes of $\mathcal H^{\otimes m}$ under multiplication by a scalar, or in other words, into quantum states of $\mathcal H^{\otimes m}$.

More specifically, this mapping encodes a subspace by a Slater determinant, such that we call the quantum state $\bm{\ket{\phi_A}}$ the \textit{determinant state} associated to the subspace $V = \operatorname{span} \left( \left\{ \ket{\phi_k} \right\}_k \right)$. We emphasize that, when thought of as an encoding of a subspace, both the global phase and the norm of this vector are irrelevant.

Formally, we have established the mapping
\begin{equation}
\label{eq:determinant_state_mapping}
\begin{cases}
\gr_m(\mathcal H) &\to S / (\mathbb C, \times) \\
V &\mapsto \left[ \bm{S_-} \ket{\phi_1} \otimes \ldots \otimes \ket{\phi_m} \right]
\end{cases}
\end{equation}
where $S$ is the set of Slater determinant vectors, and $S / (\mathbb C, \times)$ is thus the set of Slater determinants understood as quantum states, that is, defined up to a complex scalar multiple.
We have established that this mapping is well-defined, it is surjective, and we show in the appendix that it is injective (\cref{prop:injectivity}). This mapping that we introduced in the context of encoding subspaces of a Hilbert space is in fact known in mathematics as the \textit{Plücker embedding} \cite{Griffiths1994AlgebraicGeometry}.

To conclude, we have presented a one-to-one mapping of a subspace of dimension $m$ to a Slater determinant of $\mathcal H^{\otimes m}$ constructed from any basis of that subspace. 
In practice we will still encode the states $\ket{\phi_k}$ themselves, but the representation as a determinant state will serve to ensure that any expression we use is representative-independant. 
Furthermore, the fact that we can identify a subspace with a quantum state suggests that we can use expressions that are known for quantum states, and apply them on determinant states as guidance to generalize them on subspaces.

\subsection{Distance between subspaces from distance between quantum states}
\label{sec:subspace_fidelity}

Let us now see how the determinant state perspective is helpful for computing distances between subspaces.
The space of quantum states, that is, $\mathcal H / (\mathbb C, \times)$, is endowed with the Fubini-Study metric
\begin{equation}
d \left( \ket \psi, \ket \phi \right) = \arccos \left( \frac{\left| \braket{\psi | \phi} \right|}{\lVert \psi \rVert \lVert \phi \rVert} \right).
\end{equation}
We emphasize that the Fubini-Study metric is independent of the choice of representative of the quantum states, that is,
\begin{equation}
\label{eq:fs_inv}
d \left( \alpha \ket \psi, \beta \ket \phi \right) = d \left( \ket \psi, \ket \phi \right),
\end{equation}
for any $\alpha, \beta \in \mathbb C$.

It is then natural to define a Fubini-Study metric between subspaces through their respective determinant states. 
For any two $m$-dimensional subspaces $U$ and $V$ of $\mathcal H$, with respective arbitrary bases $\left\{ \ket{\psi_k} \right\}_k$ and $\left\{ \ket{\phi_k} \right\}_k$, we define
\begin{equation}
\label{eq:subspace_fidelity}
d_s \left( U, V \right) = d \left( \bm{\ket{\psi_A}}, \bm{\ket{\phi_A}} \right).
\end{equation}
First, we note that thanks to the properties of the determinant state mapping, this function is well-defined, in the sense that it does not depend on the chosen bases.
Indeed, if $\{ \ket{\tilde \psi_k} \}_k$ and $\{ \ket{\tilde \phi_k} \}_k$ are different bases of $U$ and $V$ respectively, then $d ( \bm{\ket{\psi_A}}, \bm{\ket{\phi_A}} ) = d ( \bm{\ket{\tilde{\psi}_A}}, \bm{\ket{\tilde{\phi}_A}} )$,
as a consequence of \cref{eq:det_cob} and \cref{eq:fs_inv}.

Since the Fubini-Study metric is a metric and the determinant state mapping (\cref{eq:determinant_state_mapping}) is injective, it implies that $d_s$ is itself a metric over $\gr_m (\mathcal H)$.\footnote{We remark that, in general, any choice of metric over the space of determinant states would induce a metric over subspaces.}

We can examine in more detail the behaviour of this metric. In accordance with the definition of a metric, $d_s(U,V) = 0$ if and only if $U = V$. Conversely, the maximum distance of $\pi/2$ is attained if and only if $V$ contains a non-zero vector that is orthogonal to $U$ (see \cref{prop:max_distance_subspaces}). 
More generally, we can relate this distance to the volume of the $m$-dimensional parallelotope spanned by a orthonormal basis of $V$ orthogonally projected on $U$ (see \cref{lem:det_state_fidelity}).

Having seen that the Fubini-Study metric between determinant states defines a metric between subspaces, we can reuse all the tools developed for estimating and optimizing the Fubini-Study metric between states and apply them to subspaces \cite{Sinibaldi2023Unbiasing, Gravina2024PTVMC, Havlicek2023AmplitudeRatios}.

\subsection{Operators acting on subspaces through determinant states}
\label{sec:det_state_operators}

In quantum physics, the two fundamental kinds of operators are unitary and Hermitian.
In this section, we examine their natural generalization to operators acting on subspaces.

We consider an operator $\hat{A}$ acting on $\mathcal H$ and a subspace $V$ with basis $\left\{ \ket{\phi_k} \right\}_k$, encoded through its determinant state $\bm{\ket{\phi_A}}$. 
The objective of this section is to define an operator $\bm{\hat{A}}$ acting on $\mathcal H^{\otimes m}$ such that $\bm{\hat{A}} \bm{\ket{\phi_A}}$ defines the action of $\hat{A}$ on the subspace $V$ in a meaningful way.
In this section we will be using the $\hat{\cdot}$ symbol to disambiguate operators $\hat U$ from subspaces $V$.

\paragraph{Unitary operators}
The natural way to define the action of a unitary operator $\hat{U}$ on a subspace $V$ is
\begin{equation}
\hat U V = \left\{ \hat{U} \ket{\phi}, \ket{\phi} \in V \right\}.
\end{equation}
The set $\hat U V$ is still a subspace, and if $\left\{ \ket{\phi_k} \right\}_k$ is a basis of $V$, then $\left\{ \ket{\tilde \phi_k} \right\}_k = \left\{\hat{U} \ket{\phi_k} \right\}_k$ is a basis of $\hat U V$. 
The determinant state of $\hat U V$ can thus be expressed as
\begin{equation}
\bm{\ket{\tilde \phi_A}} = \bm{S_-} \bigotimes_k \hat{U} \ket{\phi_k} 
= \hat{U}^{\otimes m} \bm{\ket{\phi_A}},
\end{equation}
since $\hat{U}^{\otimes k}$ commutes with $\bm{S_-}$. 
We deduce that the natural way to generalize a unitary operator $\hat{U}$ as an operator acting on determinant states is to define $\bm{\hat U} = \hat U^{\otimes k}$. 
We remark that $\bm{\hat U}$ is indeed a unitary operator of $\mathcal H^{\otimes m}$.

\paragraph{Hermitian operators}
Using the previously established generalization of unitaries, we will see that while unitaries are generalized by taking a tensor product, Hermitian operators are generalized by taking a sum.
Consider the family of unitary operators generated by the hermitian operator $\hat H$ 
\begin{equation}
\label{eq:infinitesimal_generator}
\hat U_a = e^{-i a \hat H}, \ a \in \mathbb R.
\end{equation}
In the previous paragraph, we have defined the generalization of $\hat U_a$ to the space of determinant states to be $\bm{\hat U_a} = \hat U_a^{\otimes m}$.
Substituting \cref{eq:infinitesimal_generator} in the definition of the generalized unitary $\bm{\hat U_a}$, we obtain, 
\begin{equation}
\bm{\hat U_a} = e^{-i a \sum_k \bm{\hat H_k}},
\end{equation}
where $\bm{\hat H_k}$ is the operator $\hat H$ acting on the $k$-th copy of the Hilbert space
\begin{equation}
\bm{\hat H_k} = \mathbb 1 \otimes \ldots \otimes \underset{\substack{\uparrow \\ \mathclap{\text{$k$-th position}}}}{\hat H} \otimes \ldots \otimes \mathbb 1.
\end{equation}
The unitary $\bm{\hat U_a}$ is therefore generated by the Hermitian operator $\sum_k \bm{\hat H_k}$, so we deduce that the natural generalization of $\hat H$ as an operator that acts on subspaces through the determinant state is $\bm{\hat H} = \sum_k \bm{\hat H_k}$.

One might question why we defined the action of a unitary operator on a subspace as $\hat U V$, but not the action of a Hermitian operator as $\hat H V$. There are three reasons. 
(i) First, unitary operators are operators that are fundamentally meant to be applied on states. The action of a unitary operator on a state defines a transformed state. 
On the other hand, the action of a Hermitian operator on a state does not fundamentally have a meaning. 
What is fundamental in the definition of Hermitian operators is their average values, which defines the expected value of a measurement.
(ii) The second justification is that unitary operators and Hermitian operators are related by the former being generated by the latter.
While our construction explicitly enforced this relationship for the generalized operators, defining the action of $\hat H$ on $V$ as $\hat H V$ would not satisfy it. In particular, the time evolution operator would no longer be generated by the Hamiltonian.
(iii) Lastly, following the same procedure as for unitary operators would have led to the operator $H^{\otimes m}$. 
While this operator is also Hermitian, dimensional analysis indicates that it is not a proper generalization of $\hat H$ since it does not have the dimension of an energy.

\subsection{Energy of a subspace and the spectrum of the Rayleigh matrix}
\label{sec:subspace_energy}

In the previous section we examined the generalization of unitary and Hermitian operators for determinant states. 
Since we have generalized the notion of a Hamiltonian to subspaces, it is now natural to examine the resulting notion of energy.

\paragraph{Energy of the determinant state}
Following the generalization of Hermitian operators in \cref{sec:det_state_operators}, the energy of a subspace $V$ spanned by a basis $\left\{ \ket{\phi_k} \right\}_k$ for a Hamiltonian $H$ is given by
\begin{equation}
\label{eq:subspace_energy}
\frac{\bm{\braket{\phi_A | H | \phi_A}}}{\bm{\braket{\phi_A | \phi_A}}}.
\end{equation}

Since $\bm H$ commutes with the antisymmetrizer $\bm S_-$, which is the orthogonal projector on the antisymmetric subspace $V_-(\mathcal H^{\otimes m})$ (see \cref{def:antisymmetric_subspace}), then $V_-(\mathcal H^{\otimes m})$ is an invariant subspace under $\bm H$.
Therefore, the restriction of $\bm{H}$ to the subspace of $\mathcal H^{\otimes m}$ spanned by the determinant states is well-defined. If we denote by $E_k$ the ordered eigenvalues of $H$ and by $\ket{E_k}$ their associated eigenstate, then the minimum energy of $\bm H$ restricted to $V_-(\mathcal H^{\otimes m})$ is
\begin{equation}
\bm{E_0} = \sum_{k=0}^{m-1} E_k.
\end{equation}
The associated eigenstate is the determinant state corresponding to the subspace $\operatorname{span} \left( \left\{ \ket{E_k} \right\}_{0 \leq k \leq m-1} \right)$.

Based on this observation, Pfau et al. \cite{Pfau2024} use \cref{eq:subspace_energy} as a cost function to optimize for the $m$ lowest excited states. Pfau et al. further show that this energy of a subspace has a simple expression that does not involve determinants,
\begin{equation}
\label{eq:det_rm_connection}
\frac{\bm{\braket{\phi_A | H | \phi_A}}}{\bm{\braket{\phi_A | \phi_A}}} = \tr \left( G^{-1} G^{(H)} \right),
\end{equation}
where $G$ is the Gram matrix of the family $\left\{ \ket{\phi_k} \right\}_k$ and $G^{(H)}$ is the projected Hamiltonian on this family,
\begin{align}
\label{eq:g_def}
G_{ij} &= \braket{\phi_i | \phi_j}, \\
\label{eq:gh_def}
G_{ij}^{(H)} &= \braket{\phi_i | H | \phi_j}.
\end{align}
We provide a detailed proof of \cref{eq:det_rm_connection} in \cref{prop:det_energy}.

The appearance of the matrix $G^{-1} G^{(H)}$ within the determinant state energy shows that this framework is in fact a reformulation of subspace methods for eigenvalue problems. In order to better examine this connection, let us first recall the theory of subspace methods for eigenvalue problems, in particular in the context of non-orthonormal bases.

\paragraph{Subspace methods and the Rayleigh matrix}
In the Lanczos method \cite{Golub, TrefethenBau}, extremal eigenvalues of a Hermitian matrix $A$ are estimated by orthonormalizing the family of $m$ Krylov vectors $\left\{ A^k v \right\}_k$ into the family $\left\{ u_k \right\}_k$, and taking the eigenvalues of the projected $m \times m$ tridiagonal matrix with coefficients $T_{ij} = u_i^\dag A u_j$.

The same procedure would also work with different families of vectors ${\phi_k}$, not necessarily orthonormal, as long as they are suitably chosen to have good overlap with the extremal eigenvectors \cite{Nightingale2001, Shixin2025Unified}.
Given a family of well-chosen linearly independent states $\ket{\phi_k}$, the eigenvalues $\mu_k$ of $G^{-1} G^{(H)}$ provide an approximation of the low-lying spectrum of the Hamiltonian $H$, where $G$ and $G^{(H)}$ are defined by \cref{eq:g_def,eq:gh_def}.\footnote{The matrix $G^{-1} G^{(H)}$ is analogous to the matrix $T$ from the Lanczos method, with the difference that the inverse of the Gram matrix $G$ accounts for the fact that the family $\left\{ \ket{\phi_k} \right\}_k$ is not necessarily orthonormal.}

The eigenvalues $\mu_k$ of $G^{-1} G^{(H)}$ only depend on the subspace $\operatorname{span} \left( \left\{ \ket{\phi_k} \right\}_k \right)$, and not on the specific choice of basis. These eigenvalues are real.\footnote{The matrix $G^{-1} G^{(H)}$ is not Hermitian in general, so it is not obvious that its eigenvalues $\mu_k$ should be real. However, this matrix is similar to $G^{- \frac 1 2} G^{(H)} G^{- \frac 1 2}$ such that they share the same eigenvalues $\mu_k$. And since $G^{- \frac 1 2} G^{(H)} G^{- \frac 1 2}$ is manifestly Hermitian, we must have $\mu_k \in \mathbb R$.}

In particular, these approximations satisfy a generalized variational principle: if we order the eigenvalues $\mu_k$ of $G^{-1} G^{(H)}$, then they are upper bounds to the corresponding ordered eigenvalues $E_k$ of the Hamiltonian $H$,
\begin{equation}
\label{eq:cit}
\forall k, \ E_k \leq \mu_k.
\end{equation}
Furthermore, the eigenvectors $\alpha^{(k)}$ of $G^{-1} G^{(H)}$ define states of the full Hilbert space
\begin{equation}
\ket{\psi^{(k)}} = \sum_p \alpha_p^{(k)} \ket{\phi_p},
\end{equation}
such that $\ket{\psi^{(k)}}$ has energy $\mu_k$. The proofs of these statements are reported in \cref{app:ritz_values}.
In \cref{app:ritz_values}, we discuss more generally why the eigenvalues $\mu_k$ of $G^{-1} G^{(H)}$ can be seen as the natural generalization of the eigenvalues of $H$ in the subspace $V = \operatorname{span}(\left\{ \ket{\phi_{\theta_k}} \right\}_k)$.

We finally remark that \cref{eq:cit} is a generalization of the ground state variational principle
\begin{equation}
\forall \ket \psi \in \mathcal H, \ E_0 \leq \frac{\braket{\psi | H | \psi}}{\braket{\psi | \psi}}.
\end{equation}
In particular, we see that $G^{-1} G^{(H)}$ appears as a generalization of the Rayleigh quotient $\braket{\psi | \psi}^{-1} \braket{\psi | H | \psi}$. 
As such, we call $G^{-1} G^{(H)}$ the \textit{Rayleigh matrix}.

From this perspective, a natural variational principle thus consists in minimizing the sum of the eigenvalues $\mu_k$ of the Rayleigh matrix, that is, $\tr ( G^{-1} G^{(H)} )$. This variational principle is known as the Ky Fan variational principle \cite{Ding2024KyFan}.\footnote{Since the trace of the Rayleigh matrix is independent of the chosen basis, we can choose an orthonormal basis $\left\{ \ket{\psi_k} \right\}_k$, such that $\tr \left( G^{-1} G^{(H)} \right) = \sum_k \braket{\psi_k | H | \psi_k}$, which is the usual form of the Ky Fan variational principle.}

\paragraph{Determinant state and Rayleigh matrix}
We now understand \cref{eq:det_rm_connection} as expressing the equivalence between the determinant state variational principle, and the Ky Fan variational principle.

Furthermore, this connection provides a better understanding of the Rayleigh matrix from the point of view of the determinant state, as the object whose eigendecomposition defines the natural generalization of the spectrum of $H$ restricted to that subspace. More practically, it provides bounds that generalize the ground state variational principle for each approximate energy $\mu_k$ found in the subspace (\cref{eq:cit}).

In the next section, we explore practical applications of the formalism that we have developed so far in this section.

\subsection{Applications of the determinant state to variational Monte Carlo methods}
\label{sec:det_state_applications}

The determinant state is also a convenient tool for practical applications. In a variational context, a subspace $V$ can be encoded by a basis of $m$ variational states $\ket{\phi_{\theta_k}}$ through the determinant state mapping, and we will see in this section that these variational subspace states can also be manipulated efficiently.

\paragraph{Sampling}
The amplitudes of the determinant state can be easily evaluated. For an element $\bm s = (s_1, s_2, \ldots, s_m)$ of the computational basis of $\mathcal H^{\otimes m}$, such that $s_k$ is an element of the computational basis of $\mathcal H$, the corresponding amplitude of $\bm{\ket{\phi_A}}$ can be expressed as
\begin{equation}
\bm{\braket{s | \phi_A}} = \frac{1}{m!} \det \Phi(\bm s),
\end{equation}
where the $m \times m$ matrix $\Phi(\bm s)$ is defined by $[\Phi(\bm s)]_{ij} = \braket{s_i | \phi_j}$ (see proof of \cref{prop:det_slater}). As such, even though $\bm{\ket{\phi_A}}$ is a linear combination of $m!$ states, its amplitudes can be queried efficiently. This also means that the probability distribution $\left| \bm{\braket{s | \phi_A}} \right|^2$ can be approximately sampled from using Markov chain Monte Carlo.\footnote{The determinant structure of the amplitudes allows for the use of rank-1 update formulas to update the determinant at a cost of $O(m^2)$ instead of $O(m^3)$ \cite{Nukala2009FastDeterminant}. While this method can increase the speed of sampling, it may lead to an accumulation of numerical errors. Since in our case the cost of evaluating determinant state amplitudes is dominated by the evaluation of $\Phi(\bm s)$ and $\Phi^{(H)}(\bm s)$ rather than the computation of the determinant, our numerical results in \cref{sec:numerical_results} were obtained without the rank-1 update method.}

\paragraph{Fubini-Study metric}
The subspace distance discussed in \cref{sec:subspace_fidelity} can be estimated using the decomposition of the Fubini-Study distance between states
\begin{multline}
\label{eq:infidelity_estimation}
d_s(U, V) = \arccos \left( \mathbb E_{\bm s \sim |\bm{\psi_A(s)}|^2} \left[ \frac{\bm{\braket{s | \phi_A}}}{\bm{\braket{s | \psi_A}}} \right] \right.\\
\left. \mathbb E_{\bm s \sim |\bm{\phi_A(s)}|^2} \left[ \frac{\bm{\braket{s | \psi_A}}}{\bm{\braket{s | \phi_A}}} \right] \right),
\end{multline}
which leads to an efficient estimator with the addition of control variates, as discussed in Ref.~\cite{Gravina2024PTVMC}.

The samples $\bm s = (s_1, \ldots, s_m)$ are configurations of the Hilbert space $\mathcal H^{\otimes m}$. 
They are sampled according to the probability distribution of $\bm{\ket{\psi_A}}$ and $\bm{\ket{\phi_A}}$ respectively, namely, $\left| \det \Psi(\bm s) \right|^2$ and $\left| \det \Phi(\bm s) \right|^2$.

One possible application of the Fubini-Study distance between subspaces is the search for approximately invariant subspaces, which could be used to simulate dynamics.

\paragraph{Approximately invariant subspace}
Let $\ket{\psi_0}$ be the initial state for dynamics. 
If we can identify an $m$-dimensional subspace $V$ that contains $\ket{\psi_0}$ and that is invariant under the Hamiltonian $H$, then the dynamics of $\ket{\psi_0}$ under $H$ can be restricted to $V$ without approximation. 
In general, however, for reasonably small $m$, we do not expect the existence of an invariant subspace containing $\ket{\psi_0}$. 

Alternatively, we can search for an \textit{approximately invariant subspace}, that is, a subspace $V$ which contains $\ket{\psi_0}$ and such that $HV = \left\{H \ket{\phi}, \ \ket{\phi} \in V\right\}$ is close enough to $V$ according to some distance, such as the one defined by \cref{eq:subspace_fidelity}.\footnote{As we discussed in \cref{sec:det_state_operators}, applying a Hermitian operator directly on states is not a natural operation. Instead of requiring $V_\theta$ to be close to the subspace $HV_\theta$, we could instead consider the subspace $e^{-i H \delta} V_\theta$ for some $\delta \in \mathbb R$. While both choices would lead to the optimization problem minimizing at an invariant subspace in the case that it exists, it is not yet clear what the minimum would be in either cases if an invariant subspace does not exist.}
To solve this problem, we can variationally encode a subspace $V_\theta = \left\{ \ket{\psi_0}, \ket{\psi_{\theta_1}}, \ldots, \ket{\psi_{\theta_{m-1}}} \right\}$ and solve the optimization problem
\begin{equation}
\theta^* = \operatorname{arg\,min}_\theta \ d_s(V_\theta, HV_\theta).
\end{equation}
The loss function of this optimization problem can be estimated through \cref{eq:infidelity_estimation}, and estimation of its gradient has also been studied \cite{Gravina2024PTVMC}.

Once an approximately invariant subspace $V^*$ has been obtained, we simply need to perform the dynamics in that subspace, which is precisely the object of \cref{sec:bridge}.

\paragraph{Energy estimation}
It is also straightforward to estimate the energy of a subspace (\cref{eq:subspace_energy}) using the standard decomposition
\begin{equation}
\frac{\bm{\braket{\phi_A | H | \phi_A}}}{\bm{\braket{\phi_A | \phi_A}}} = \mathbb E_{\bm s \sim |\bm{\phi_A(s)}|^2} \left[ \frac{\bm{\braket{s | H | \phi_A}}}{\bm{\braket{s | \phi_A}}} \right].
\end{equation}
As shown in \cref{sec:subspace_energy}, the minimization of this scalar quantity provides the subspace spanned by the $m$ eigenstates of $H$ with lowest energy, but extracting the individual energies and the corresponding states requires computing the entire Rayleigh matrix $G^{-1} G^{(H)}$.

In \cref{sec:subspace_energy} we argued that the Rayleigh quotient of a single state ${\braket{\psi | H | \psi}}/{\braket{\psi | \psi}}$ can be generalised to the Rayleigh matrix of a subspace $G^{-1}G^{(H)}$.
Similarly, the Monte Carlo estimator for the Rayleigh quotient
\begin{equation}
\frac{\braket{\psi | H | \psi}}{\braket{\psi | \psi}} = \mathbb{E}_{s \sim |\psi(s)|^2} \left[ \frac{\braket{s | H | \psi}}{\braket{s | \psi}} \right],
\end{equation}
is generalised to the following estimator for the Rayleigh matrix (due to Pfau et al. in Ref.~\cite{Pfau2024}, see proof in \cref{prop:rm_estimator})
\begin{equation}
\label{eq:rm_decomposition}
G^{-1} G^{(H)} = \mathbb{E}_{\bm s \sim |\bm{\phi_A(s)}|^2} \left[ \Phi(\bm s)^{-1} \Phi^{(H)}(\bm s) \right],
\end{equation}
where the matrices $\Phi(\bm s)$ and $\Phi^{(H)}(\bm s)$ are defined by 
\begin{align}
\left[ \Phi(\bm s) \right]_{ij} &= \braket{s_i | \phi_j}, \\
\left[ \Phi^{(H)}(\bm s) \right]_{ij} &= \braket{s_i | H | \phi_j}.
\end{align}
The estimator of the Rayleigh matrix also generalizes the desirable zero-variance property of the estimator of the energy: if the subspace is an invariant subspace under $H$ (and thus contains $m$ proper eigenvectors of $H$), then the local energy matrix $\Phi(\bm s)^{-1} \Phi^{(H)}(\bm s)$ does not depend on $\bm s$ (\cref{prop:zero-variance}).

\paragraph{Extracting states from subspaces}
In the excited states estimation method discussed above, estimating the Rayleigh matrix was required to extract the eigenvalue and eigenstate approximations from the optimized subspace. This process can be extended to any subspace, irrespective of how its basis states were obtained.

In the context of ground state optimization, the solution obtained by variational Monte Carlo can be improved by extracting the linear combination from a basis consisting of the optimization result and other states.
Example of states that can be added to the basis include states reached earlier during the optimization, or Krylov states built from the final optimization state \cite{Wang2025Lanczos, Chen2022Lanczos, Sorella2001LanczosNonsense}.

Another possible use of extracting optimal linear combinations out of a subspace from the Rayleigh matrix is the generalization of ground states optimized for given values of the Hamiltonian parameters \cite{Duguet2024, Rath2025}. Consider a parametric Hamiltonian of the form
\begin{equation}
H(\gamma) = \sum_p \gamma_p H_p.
\end{equation}
Suppose that we have obtained $m$ states $\ket{\phi_k} = \ket{\phi_{\gamma_k}}$ that approximate the ground state at $m$ different values of the Hamiltonian parameter $\gamma$. Exploiting the continuity of eigenstates within a phase, we can approximate the ground state at other values of $\gamma$ as linear combinations of the states $\ket{\phi_k}$. The Rayleigh matrix at any parameter $\gamma$ is given by
\begin{equation}
\label{eq:rm_linearity}
G^{-1} G^{(H(\gamma))} = \sum_p \gamma_p G^{-1} G^{(H_p)}.
\end{equation}
Therefore, it can be estimated at any value $\gamma$ with a single estimation of each Rayleigh matrices $G^{-1} G^{(H_p)}$, using either the determinant state estimator (\cref{eq:rm_decomposition}) or the sum of states estimator (\cref{eq:sum_of_states_1,eq:sum_of_states_2}). Then, for any value $\gamma$, we simply need to diagonalize the small $m \times m$ Rayleigh matrix $G^{-1} G^{(H(\gamma))}$ in order to obtain its lowest eigenvalue $\mu^{(\gamma)}$ and the corresponding eigenvector $\alpha^{(\gamma)}$. It defines a state of the full Hilbert space
\begin{equation}
\ket{\psi^{(\gamma)}} = \sum_k \alpha_k^{(\gamma)} \ket{\phi_k}
\end{equation}
that approximates the ground state of $H(\gamma)$, and has energy $\mu^{(\gamma)}$.

We illustrate this method on a simple example in \cref{fig:carpet}. The Hamiltonian is the transverse-field Ising model defined in \cref{eq:tfim} on an open chain with 12 sites. We choose the open chain in order to decrease the number of symmetries and ensure that the ground states do not belong to a subspace whose dimension is much smaller than that of the Hilbert space.

\begin{figure}[ht]
\centering
\includegraphics[width=\linewidth]{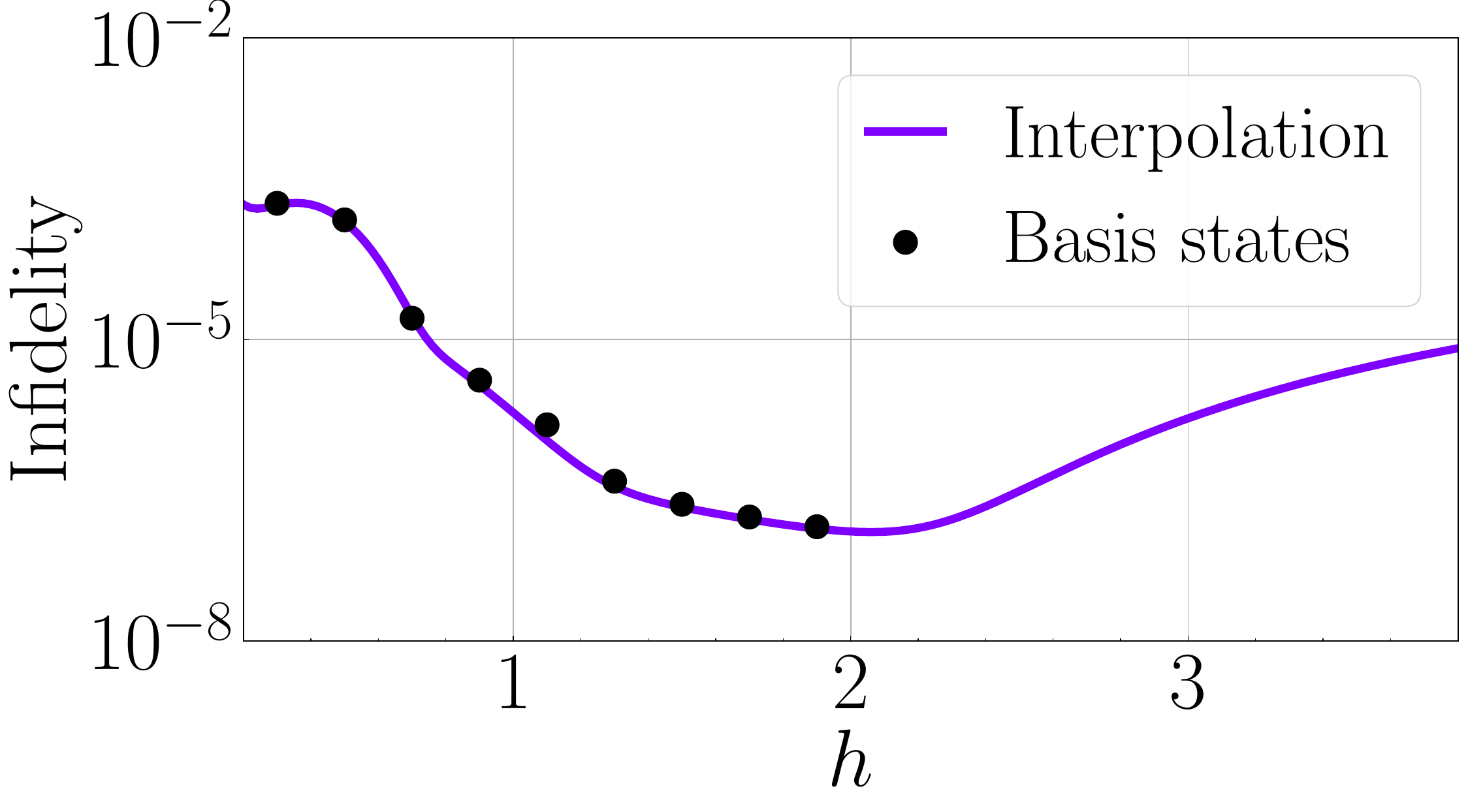}
\caption{Infidelity between the exact ground state and estimations of the ground state for various values of the transverse-field Ising Hamiltonian parameter $h$. The parameter $J$ is fixed to $J=1$. The Hamiltonian is defined on an open chain of 12 sites. The black dots indicate nine ground state approximations obtained for various values of $h$. These variational states are restricted Boltzmann machines \cite{Carleo2017} with $\alpha=4$. The purple line shows the precision of the ground state interpolation obtained by taking linear combination of the nine ground state approximations.}
\label{fig:carpet}
\end{figure}

We see that the linear combination states successfully generalize to unseen values of $\gamma$, both for interpolation and, to some degree, extrapolation.

The average value of any operator $A$ on these linear combination states can then be conveniently evaluated using the Rayleigh matrix of $A$. For a given value of $\gamma$, let $p$ be the index of the lowest eigenvalue of $G^{-1} G^{(H(\gamma))}$, and let $P_\gamma$ denote the matrix of eigenvectors of $G^{-1} G^{(H(\gamma))}$. Then, under the assumption that the eigenvalues of $G^{-1} G^{(H(\gamma))}$ are all distinct, the average value of $A$ on the ground state approximation at $\gamma$ is given by
\begin{equation}
\frac{\braket{\psi^{(\gamma)} | A | \psi^{(\gamma)}}}{\braket{\psi^{(\gamma)} | \psi^{(\gamma)}}} = \left[ P_\gamma^{-1} G^{-1} G^{(A)} P_\gamma \right]_{pp}.
\end{equation}
A proof is reported with \cref{prop:ritz_vector_average}.

\begin{table*}[ht]
\centering
\begin{tblr}{hlines, vlines, colspec={c c | c c}, row{1}={font=\bfseries}, row{1-Z}={font=\footnotesize}, rowsep=0.2cm}
 & Section & Single state & Subspace \\
\hline
Original object & \ref{sec:det_state_mapping} & 1-d subspace with basis $\left\{ \ket \psi \right\}$ & $m$-d subspace with basis $\left\{ \ket{\phi_k} \right\}_k$ \\
Quantum state & \ref{sec:det_state_mapping} & $\ket \psi$ & $\bm{\ket{\phi_A}} = \bm{S_-} \ket{\phi_1} \otimes \ldots \otimes \ket{\phi_m}$ \\
Hilbert space & \ref{sec:det_state_mapping} & $\mathcal H$ & $\mathcal H^{\otimes m}$ \\
Distance & \ref{sec:subspace_fidelity} & $\displaystyle d \left( \ket \psi, \ket \phi \right) = \arccos \left( \frac{\left| \braket{\psi | \phi} \right|}{\lVert \psi \rVert \lVert \phi \rVert} \right)$ & $d_s(U, V) = d(\bm{\ket{\psi_A}}, \bm{\ket{\phi_A}})$ \\
Unitary operator & \ref{sec:det_state_operators} & $U$ & $\bm U = U^{\otimes m}$ \\
Hermitian operator & \ref{sec:det_state_operators} & $H$ & $\bm H = \sum_k \bm{H_k}$ \\
Energy & \ref{sec:subspace_energy} & $\displaystyle \frac{\braket{\psi | H | \psi}}{\braket{\psi | \psi}}$ & $\displaystyle \frac{\bm{\braket{\phi_A | H | \phi_A}}}{\bm{\braket{\phi_A | \phi_A}}} = \tr \left( G^{-1} G^{(H)} \right)$ \\
Variational principle & \ref{sec:subspace_energy} & $E_0 \leq \displaystyle \frac{\braket{\psi | H | \psi}}{\braket{\psi | \psi}}$ & $\displaystyle \forall k, \ E_k \leq \lambda_k \left( G^{-1} G^{(H)} \right)$ \\
Energy estimator & \ref{sec:det_state_applications} & $\displaystyle \frac{\braket{\psi | H | \psi}}{\braket{\psi | \psi}} = \mathbb{E}_{s \sim |\psi(s)|^2} \left[ \frac{\braket{s | H | \psi}}{\braket{s | \psi}} \right]$ & $\displaystyle G^{-1} G^{(H)} = \mathbb{E}_{\bm s \sim |\bm{\phi_A(s)}|^2} \left[ \Phi(\bm s)^{-1} \Phi^{(H)}(\bm s) \right]$ \\
\end{tblr}
\caption{Summary of the generalisation of quantum states to subspaces, as explored in \cref{sec:determinant_state}.}
\label{tab:summary}
\end{table*}

\section{Bridge: improving variational dynamics accuracy with subspace methods}
\label{sec:bridge}
Various methods have been proposed for the variational approximation of quantum dynamics, including time-dependent variational Monte Carlo \cite{Carleo2017Boson, Schmitt2020MarkusMarkus}, projected-time variational Monte Carlo (p-tVMC) \cite{Sinibaldi2023Unbiasing, Gravina2024PTVMC}, and global-in-time dynamics \cite{Walle2024Global, Sinibaldi2024Galerkin}.
In this section, we discuss how the subspace spanned by approximate Arnoldi states obtained by one of the above-mentioned methods can be used to extract linear combinations that improve the precision of the approximate dynamics.
In the following, we will refer to this approach as \textit{Bridge}, where the Arnoldi basis states are seen as pillars on which a bridge is built.
We remark that the use of linear combinations of Arnoldi states for extrapolating  dynamics has already been examined for exact dynamics~\cite{Minganti2022FasterThanTheClockLanczos}. 
A related strategy has recently been discussed in Ref.~\cite{Sinibaldi2024Galerkin} for global-in-time variational dynamics. We comment on the relation to this work at the end of \cref{sec:estimation_rayleigh_matrix}.

We will see that Bridge is a computationally inexpensive method that can significantly mitigate the errors arising from the time discretization of the method that produced the approximate Arnoldi basis states.

\paragraph{}
Consider a set of $m$ variational states $\left\{ \ket{\phi_{\theta_k}} \right\}_{k \in \{ 0, \ldots, m-1 \}}$ that approximate the exact dynamics at times $t = k\delta$ of a system evolving according to the Hamiltonian $H$ with initial state $\ket{\phi_{\theta_0}}$.
These states may have been produced by t-VMC, p-tVMC, tensor network methods or any other variational approach, such that their amplitudes can be queried efficiently.
Since these states define a low-dimensional subspace, we can identify the time-dependent linear combination 
\begin{equation}
\ket{\psi_{\alpha(t)}} = \sum_k \alpha_k(t) \ket{\phi_{\theta_k}}, \ \alpha(t) \in \mathbb C^m
\end{equation}
which solves the Schrödinger equation within the variational space $\operatorname{span}(\left\{ \ket{\phi_{\theta_k}} \right\}_k)$, and take it as a new and possibly more accurate approximation of the dynamics.
Since at time $t = k\delta$ the linear variational state can represent $\ket{\phi_{\theta_k}}$, the optimal state in the subspace $\operatorname{span}(\left\{ \ket{\phi_{\theta_k}} \right\}_k)$ has at least the accuracy of the original basis states, but it could achieve a better approximation of the dynamics. Although it is not guaranteed, we empirically observe that the linear combinations obtained by this procedure are always close to the optimal linear state, and thus never underperforms compared to the original states $\ket{\phi_{\theta_k}}$.

Importantly, this approach not only improves on previously computed solutions at times $t = k\delta$, but also allows interpolating between times $t=k\delta$ and extrapolating beyond $t=(m-1)\delta$.

\paragraph{}
In the next section, we will discuss how the parameter $\alpha(t)$ is determined through the time-dependent variational principle.
We emphasize that during this optimization, the parameters $\theta_k$ remain fixed and only the coefficients $\alpha_k$ are varied.
As such, we simplify notations by denoting $\ket{\phi_k} = \ket{\phi_{\theta_k}}$.
We assume that the family $\left\{ \ket{\phi_k} \right\}_k$ is linearly independent, which should always be the case in practice.\footnote{If it is not the case, then the number of states in the family can be reduced until the family is linearly independent. Linear dependence can be immediately noticed when Monte Carlo sampling the Rayleigh matrix with the determinant state estimator (\cref{eq:rm_decomposition}), since the probability of all samples is zero in this case.}

\subsection{Time-dependent variational principle for a linear state}

To determine the time-dependent parameter vector $\alpha(t)$, we use the time-dependent variational principle (TDVP). The linear structure of $\ket{\psi_\alpha}$ greatly simplifies the TDVP, resulting in the equation
\begin{equation}
\label{eq:bridge}
\dot \alpha = - i G^{-1} G^{(H)} \alpha.
\end{equation}
The direct application of the TDVP commonly used with neural quantum states where the norm and global phase are not constrained would yield a different equation. However, as discussed in \cref{app:variational_principle}, it is in fact equivalent to \cref{eq:bridge}.

This equation is an effective Schrödinger equation, where $\alpha$ evolves under the action of an effective Hamiltonian given by the \textit{Rayleigh matrix} $G^{-1} G^{(H)}$, previously discussed in \cref{sec:subspace_energy}. The solution of this equation is given by
\begin{equation}
\label{eq:linear_dynamics}
\alpha(t) = e^{- i G^{-1}G^{(H)} t} \alpha(0).
\end{equation}

We recall from \cref{sec:subspace_energy} that the spectrum of the Rayleigh matrix is real.
Consequently, the evolution of $\alpha(t)$ is similar to an evolution under a Hermitian Hamiltonian, where the eigencomponents rotate and the dynamics never converges.

Crucially, unlike the general case, the TDVP for a linear state only requires the evaluation of $G^{-1} G^{(H)}$ once, since it does not depend on $\alpha$. From one evaluation of $G^{-1} G^{(H)}$, we can obtain the solution $\alpha(t)$ from \cref{eq:linear_dynamics} at any times for a negligible cost.

This suggests that Bridge can be used both to improve the precision at times $t = k \delta$, and to interpolate between these times, yielding a truly continuous approximation of the dynamics.

The numerical implementation of Bridge thus consists in two steps. (i) Estimate the Rayleigh matrix. (ii) Obtain $\alpha(t)$ from \cref{eq:linear_dynamics}.

The first step requires an estimator for the Rayleigh matrix, which we discuss in the next section.

The second step is simple, since the Rayleigh matrix is a small $m \times m$ matrix that can be exponentiated exactly at a negligible cost. In practice, the basis states $\ket{\phi_k}$ can be close to linear dependence, which causes instability in the exact exponentiation of the Rayleigh matrix estimate. To circumvent this problem, we perform these operations in arbitrary precision.

\paragraph{Time-dependent Hamiltonian}
We note that while we focus on the case of a time-independent Hamiltonian, Bridge can also be used to tackle time-dependent problems. 
If the Hamiltonian has a time dependence of the form
\begin{equation}
H(t) = \sum_p \beta_p(t) H_p,
\end{equation}
then it is enough to estimate the Rayleigh matrix of each term $H_p$ once to obtain the Rayleigh matrix of the Hamiltonian at all times (see \cref{eq:rm_linearity}).

\subsection{Estimation of the Rayleigh matrix}
\label{sec:estimation_rayleigh_matrix}

One way of estimating the Rayleigh matrix $G^{-1} G^{(H)}$ is the \textit{determinant state estimator} discussed in the previous section (\cref{eq:rm_decomposition}). 
Another way, proposed by Sinibaldi et al. in Ref.~\cite{Sinibaldi2024Galerkin}, relies on individually estimating $G$ and $G^{(H)}$ up to a common constant that cancels out in the product $G^{-1} G^{(H)}$. 
This estimator is based on the decomposition
\begin{align}
\label{eq:sum_of_states_1}
\frac{1}{\lVert p \rVert} G_{ij} &= \sum_s p(s) \frac{\braket{\phi_i | s} \braket{s | \phi_j}}{\sum_k \left| \braket{s | \phi_k} \right|^2}, \\
\label{eq:sum_of_states_2}
\frac{1}{\lVert p \rVert} G^{(H)}_{ij} &= \sum_s p(s) \frac{\braket{\phi_i | s} \braket{s | H | \phi_j}}{\sum_k \left| \braket{s | \phi_k} \right|^2},
\end{align}
where samples are taken from the distribution
\begin{equation}
\label{eq:sum_of_states_prob}
p(s) = \frac{1}{\lVert p \rVert} \sum_k \left| \braket{s | \phi_k} \right|^2,
\end{equation}
and $\lVert p \rVert = \sum_s \sum_p \left| \braket{s | \phi_p} \right|^2$ is a normalization constant.
Once $G / \lVert p \rVert$ has been obtained, it must be inverted in order to obtain the Rayleigh matrix. 
As $G$ is typically ill-conditioned, we need to use the pseudo-inverse with a well-chosen singular value cut-off.
Since this estimator relies on samples drawn from the summed probability distributions of all states $\ket{\phi_k}$, we call it the \textit{sum of states} (SOS) estimator.

\begin{table}[t]
\centering
\begin{tabular}{| c || c | c |}
\hline
& Determinant & Sum \\
\hline\hline
Sampling & $O \left( m \gamma + m^3 \right)$ & $O \left( m \gamma \right)$ \\
\hline
Local term & $O \left( m^2 l \gamma + m^3 \right)$ & $O \left( m l \gamma + m^2 \right)$ \\
\hline
\end{tabular}
\caption{Cost of the determinant state and sum of states estimators, where $\gamma$ is the complexity of evaluating the network of one of the basis states, $l$ is the number of connected elements of the Hamiltonian, and $m$ is the number of basis states.
\textit{Sampling} refers to the complexity of querying the amplitude of a sample, under the assumption that the new sample proposal only changes the configuration in a single copy of the original Hilbert space $\mathcal H$ (but we do not use the rank-1 update determinant formula).
\textit{Local term} refers to the complexity of evaluating the local term for a single sample. It is either the local energy matrix $\bm\Phi(\bm s)^{-1}\bm\Phi^{(H)}(\bm s)$ or the local matrix $\frac{\braket{\phi_i | s} \braket{s | H | \phi_j}}{\sum_k \left| \braket{s | \phi_k} \right|^2}$.}
\label{tab:estimator_costs}
\end{table}

In \cref{tab:estimator_costs}, we report the computational complexity of querying the probability distribution and local estimator for each estimator.
While the determinant state estimator is more expensive per sample, numerical experiments reported in \cref{fig:estimator_difference_small} indicate that it leads to more accurate results. To evaluate the performance of each estimator, we run Bridge with a given basis of states for a single quench of the $4 \times 4$ transverse-field Ising model, which is discussed in more details in \cref{sec:numerical_results}.
We run Bridge with each Rayleigh matrix estimator on a single A100 GPU with an increasing number of samples.
We find that the best performance is obtained with the determinant state estimator, and so we will be using this estimator for the numerical results of \cref{sec:numerical_results}.
The importance of using the determinant state estimator will be further exemplified with \cref{fig:various_delta_sos}, which shows that the performance of the sum of states estimator is inconsistent and highly dependent on the choice of the pseudo-inverse singular value cut-off.
In \cref{app:rm_estimators}, we provide further details regarding the various possible estimators of the Rayleigh matrix and their associated benchmarks, as well as theoretical justification for the superiority of the determinant state estimator. In particular, we justify that the determinant state estimator is more resilient to basis states that are near linear dependence, which is the case for approximate dynamics states such as those obtained by p-tVMC, especially as the size of the basis increases.

\begin{figure}[t]
\centering
\includegraphics[width=\columnwidth]{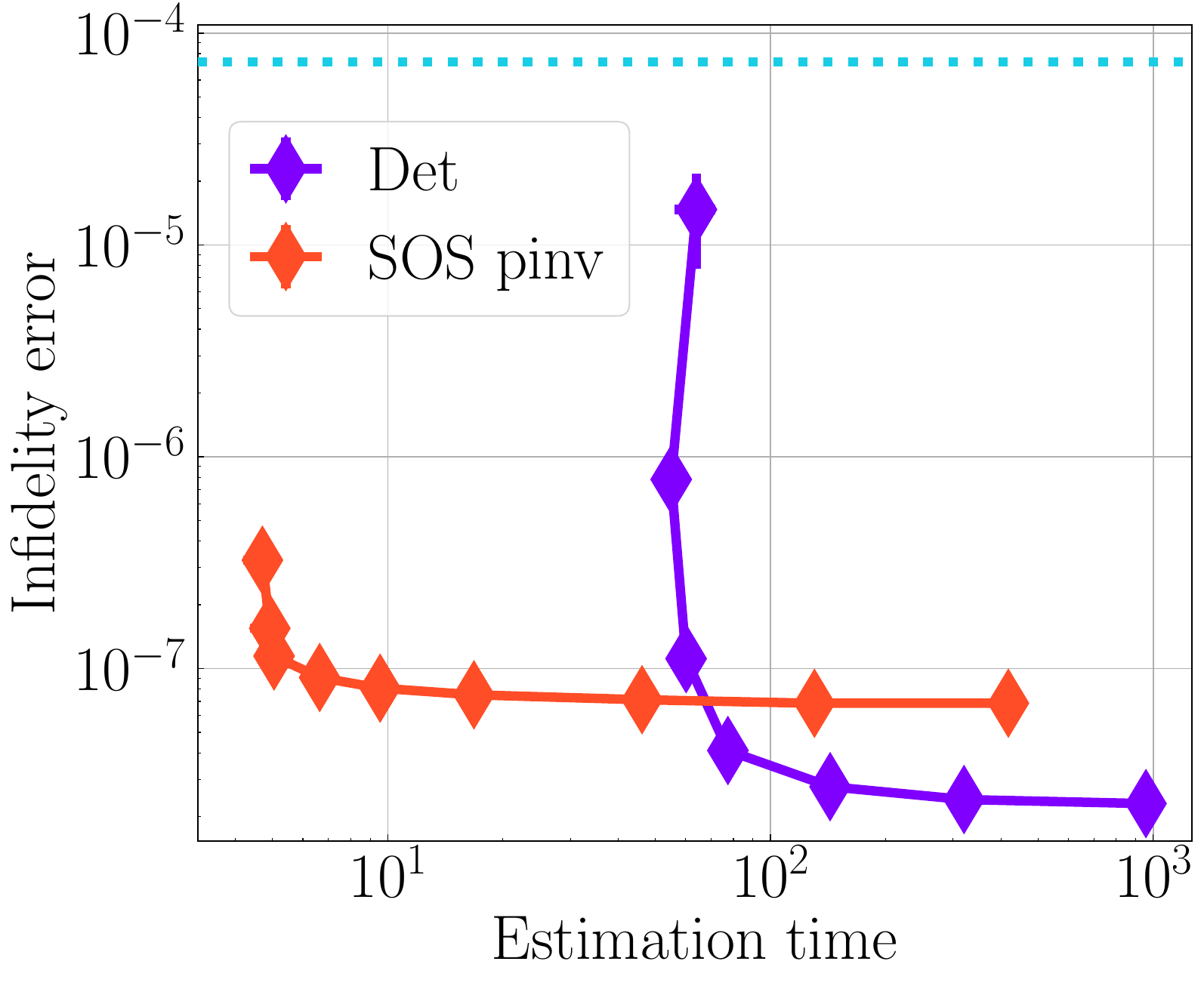}
\caption{Comparison between the final infidelity error of Bridge performed using two estimators of the Rayleigh matrix as a function of the estimation time of the Rayleigh matrix, corresponding to varying numbers of samples. The two estimators are: in purple, the determinant state estimator of \cref{eq:rm_decomposition}, and in red, the sum of states (SOS) estimator of \cref{eq:sum_of_states_1,eq:sum_of_states_2} where $G$ is inverted using the pseudo-inverse with an appropriately chosen singular value cut-off. The dotted horizontal blue line indicates the performance of the original states $\ket{\phi_k}$. The Hamiltonian is the transverse-field Ising Hamiltonian on a $4 \times 4$ lattice with $(h,J) = (1,0.1)$, and the basis states are 136 CNN states obtained with p-tVMC from time 0 to 1.35 with $\delta=0.01$.}
\label{fig:estimator_difference_small}
\end{figure}

\paragraph{Relation to Ref.~\cite{Sinibaldi2024Galerkin}}
In Ref.~\cite{Sinibaldi2024Galerkin}, the authors employed a similar method of variational optimization of linear coefficients to post-process global-in-time dynamics, using the sum of states estimator. In their specific case, the basis states are not approximate dynamics states, and therefore do not necessarily suffer from near linear dependence. Additionally, the size of their bases are small, never exceeding 20 states, which might explain why the sum of states estimator was sufficient. 
However, as detailed in \cref{app:rm_estimators}, bases of states obtained by the time-dependent variational principle or p-tVMC are typically near linear dependence, which causes the sum of states estimator to fail, such that we instead rely on the more resilient determinant state estimator.

\subsection{Estimating observables}
\label{sec:estimate_observables}

Once the parameters $\alpha(t)$ have been obtained, the state $\ket{\psi_{\alpha(t)}} = \sum_k \alpha_k(t) \ket{\phi_k}$ can be sampled from and average values of observables can be estimated using standard methods.

Alternatively, the average value of an observable $A$ can be estimated for any number of times from a single estimation. Due to the linearity of the variational state, its average values can be expressed as
\begin{equation}
\frac{\braket{\psi_{\alpha(t)} | A | \psi_{\alpha(t)}}}{\braket{\psi_{\alpha(t)} | \psi_{\alpha(t)}}} = \frac{\alpha^\dag (t) G^{(A)} \alpha(t)}{\alpha^\dag (t) G \alpha(t)}.
\end{equation}
The matrices $G$ and $G^{(A)}$ can be estimated using the estimator of \cref{eq:sum_of_states_1,eq:sum_of_states_2}. We remark that the determinant state estimator cannot be used here, since it does not provide $G$ and $G^{(A)}$ independently.

While the second method of estimating average values should be more efficient when computing many data points, we did not find a simple method to estimate statistical error bars. As such, for the numerical results, we will restrict ourselves to the first method.\footnote{We note that the method for estimating average values relying only on Rayleigh matrices detailed in \cref{prop:ritz_vector_average} does not apply here, since it relies on the orthogonality of the Ritz vectors. Here, the states $\ket{\psi_{\alpha(t)}}$ are linear combinations of Ritz vectors, and since the norms of the Ritz vectors are unknown, it is not clear how to complete $\ket{\psi_{\alpha(t)}}$ into an orthogonal basis.}

\section{Numerical results}
\label{sec:numerical_results}

We study the performance of Bridge for the dynamics of the transverse-field Ising model
\begin{equation}
\label{eq:tfim}
H = - J \sum_{\braket{i,j}} \sigma_i^{(z)} \sigma_j^{(z)} - h \sum_i \sigma_i^{(x)},
\end{equation}
where we examine in particular the time evolution of the average magnetization
\begin{equation}
M_x = \frac{1}{n} \sum_k \sigma_k^{(x)}.
\end{equation}
We consider square lattices of different sizes, with periodic boundary conditions for $J=1$ and various values of $h/h_c$ with $h_c = 3.044J$, the value of $h$ at the phase transition. The initial state is the uniform state
\begin{equation}
\ket{\psi_0} = \bigotimes_{k=1}^n \frac{\ket \uparrow + \ket \downarrow}{\sqrt 2}.
\end{equation}
To obtain the original dynamics states $\ket{\phi_k}$, we use p-tVMC \cite{Sinibaldi2023Unbiasing, Gravina2024PTVMC, Nys2024}. More details are given in \cref{app:ptvmc}.
The networks used in this section are either convolutional neural networks \cite{Schmitt2020MarkusMarkus} (for $4\times 4$ lattices) or vision transformers \cite{Viteritti2023a} (for $8 \times 8$ lattices). Both networks enforce the translational invariance of the dynamics. 
For more details on the architectures, we refer to appendix H of Ref.~\cite{Gravina2024PTVMC}.

\begin{figure*}[t]
\centering
\includegraphics[width=\linewidth]{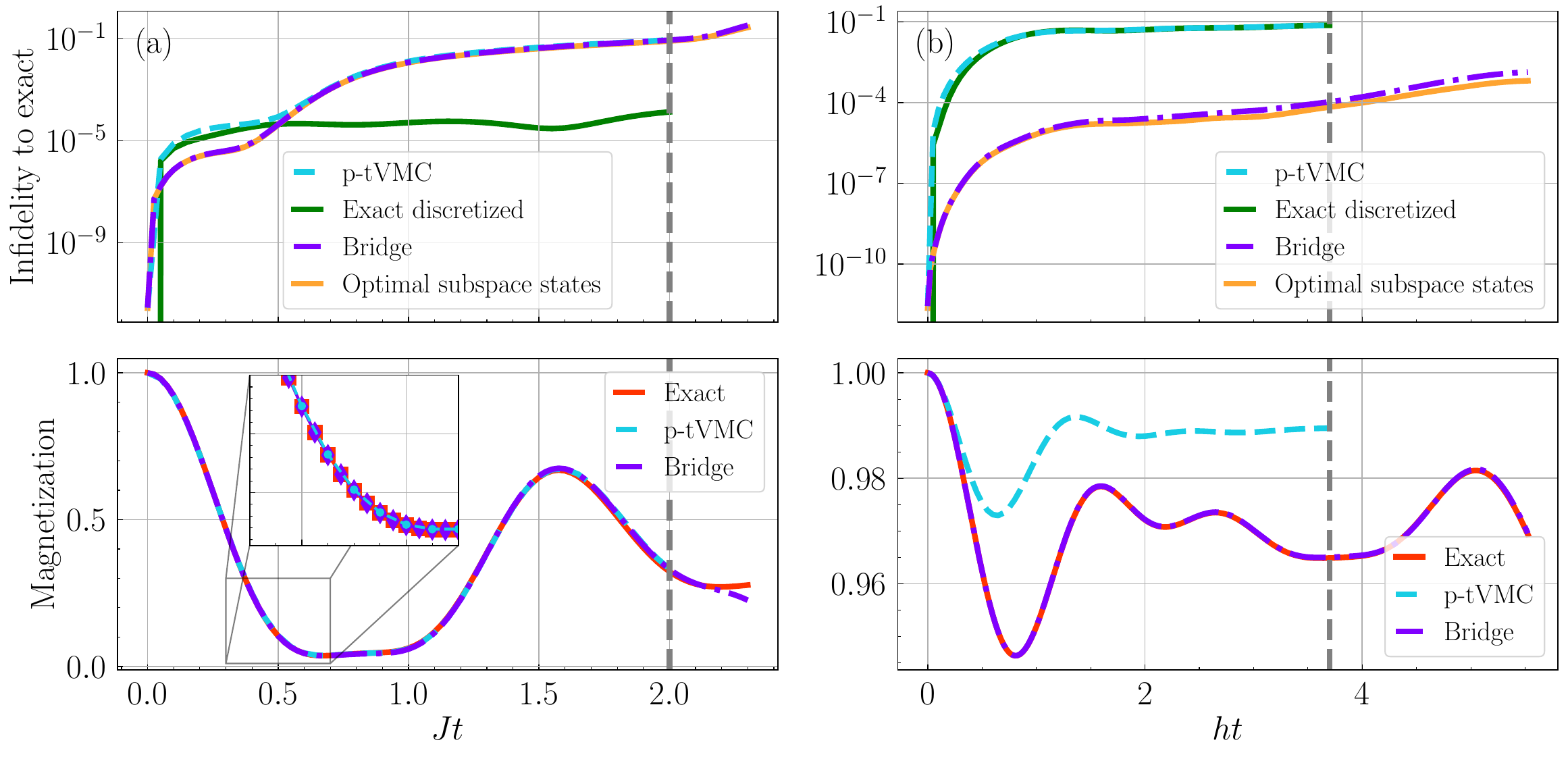}
\caption{Infidelity to exact dynamics and magnetization on the $4 \times 4$ transverse-field Ising model with $h=0.1h_c$ (a) and $h=2h_c$ (b). The original states (blue) are convolutional neural networks optimized with p-tVMC using the scheme SLPE2 with time step $J\delta = 0.05$ (a) and $h\delta = 0.05$ (b). For the Bridge states (purple), the Rayleigh matrix is estimated with the determinant state estimator using 3000 samples. Exact discretized (green) refers to the dynamics resulting from the exact application of the discretized time evolution operator. The vertical grey line highlights the final time of the p-tVMC states.
}
\label{fig:bridge}
\end{figure*}

In \cref{fig:bridge}, we consider a $4 \times 4$ lattice, where we have access to the exact dynamics, such that we can measure the performance of the various methods by evaluating its infidelity with the exact dynamics.
We run Bridge on the states $\ket{\phi_k}$ obtained using p-tVMC with the discretization scheme SLPE2 \cite{Gravina2024PTVMC}.
In the figure, the light blue line corresponds to the p-tVMC data while the purple line corresponds to Bridge.
We examine the two cases $h = 0.1h_c$ (left) and $h = 2h_c$ (right). 
In the first case, we see that Bridge improves on p-tVMC by one order of magnitude up to time $0.5J$ and then does not improve any more. 
In the second case, we see that Bridge significantly improves on p-tVMC, for all times, by 3 or 4 orders of magnitudes. 
In yellow, we plot the optimal infidelity with respect to the exact dynamics that can be achieved within the subspace spanned by the original states $\ket{\phi_k}$, which we observe Bridge is very close to.

To understand when Bridge improves on the states $\ket{\phi_k}$, we detail the two sources of error arising from the p-tVMC method used to obtain them (see \cref{app:ptvmc} for more details).
(i)~The first is \textit{discretization error} emerging from the low-order expansion of the unitary evolution.
(ii)~The second is \textit{optimization error}, which arises because the optimization problem at each discrete time-step cannot be performed exactly, for example because of finite variational expressivity.
In \cref{fig:bridge}, we present in green the exact (non-variational) discretized dynamics, that is, the dynamics when applying exactly the operator resulting from the Taylor expansion of the time evolution operator at each step. 
When the p-tVMC infidelity is close to the discretized infidelity, this means that the p-tVMC error is dominated by discretization error. 
Conversely, when the p-tVMC infidelity is above the discretized infidelity, the p-tVMC error is dominated by the optimization error. 
What we therefore identify in \cref{fig:bridge} is that Bridge improves the performance of the original states when they are dominated by discretization error, but not when they are dominated by optimization error.

This behaviour can be understood by looking at the form of each error. 
At each time step, the discretization error is the remainder of the Taylor expansion formula, which is a linear combination of powers of $H$ applied to the current state.\footnote{In the case of SLPE schemes, the Hamiltonian is split between the diagonal part $H_0$ and the off-diagonal part $H_1$, such that the remainder would instead be a linear combination of products of $H_0$ and $H_1$.} On the other hand, optimization error has no particular structure, which justifies why linear combinations cannot significantly correct it.

In the previous section, we saw that the special structure of the time-dependent variational principle allows the evaluation of the linear parameter $\alpha$ for many times at almost no additional cost, meaning that Bridge can also be used to interpolate between time steps.
The insets of \cref{fig:bridge} show that Bridge is indeed successful at interpolating between the times corresponding to the original states $\ket{\phi_k}$.

Lastly, \cref{fig:bridge} also examines whether Bridge can extrapolate beyond the times of the $\ket{\phi_k}$. 
In both cases, we see that while Bridge does have some extrapolating power, its performance predictably wanes over time.

\begin{figure*}
\centering
\includegraphics[width=\linewidth]{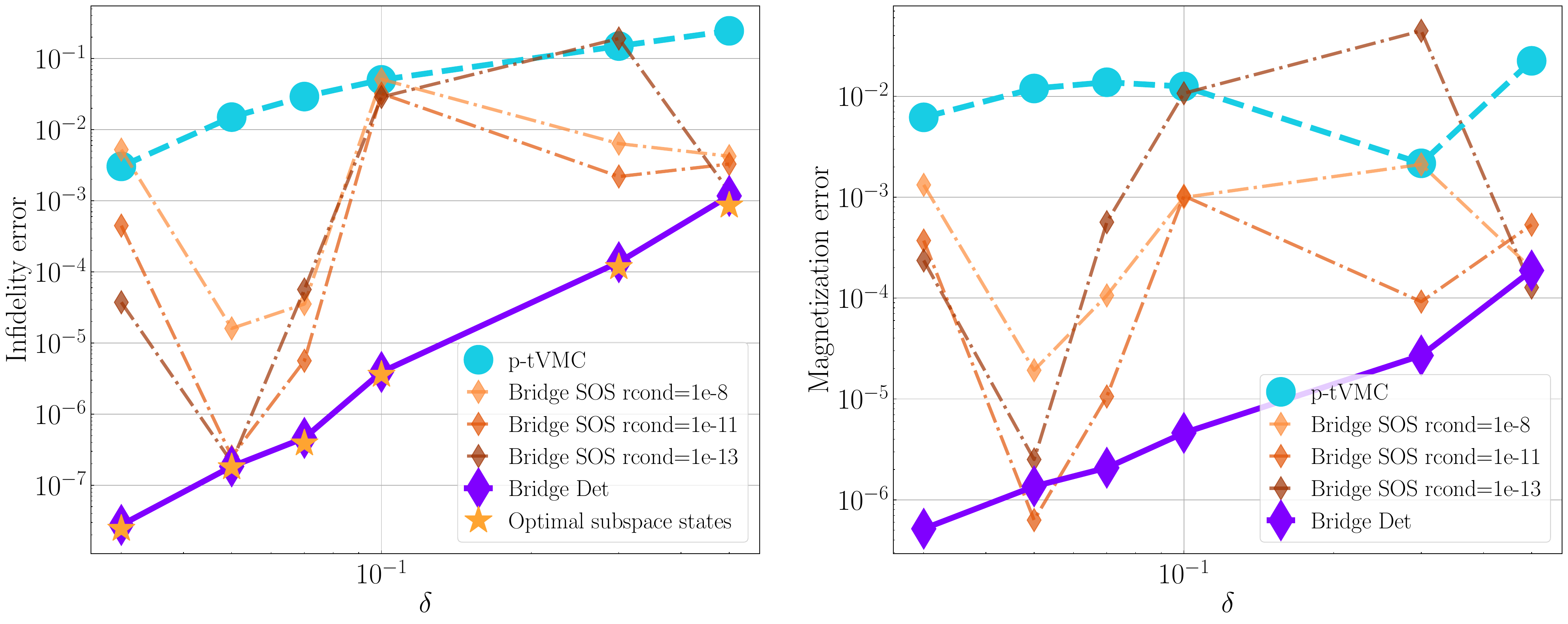}
\caption{Performance of Bridge with either the determinant state estimator (Det) or sum of states estimator (SOS), using families of p-tVMC states generated with varying time steps $\delta$ between time 0 and 1, for the transverse-field Ising Hamiltonian on a $4 \times 4$ lattice with $(h,J) = (1,0.1)$. The performance of Bridge with the sum of states estimator is shown for various values of the pseudo-inverse singular value cut-off \texttt{rcond} for inverting $G$.}
\label{fig:various_delta_sos}
\end{figure*}

\paragraph{}
\Cref{fig:various_delta_sos} shows the performance of Bridge in the case $(h,J) = (1, 0.1)$ using convolutional neural network states generated with p-tVMC, with varying time steps $\delta$. In this case, Bridge is performed with both the determinant state and the sum of states estimator, for various values of the pseudo-inverse singular value cut-off. We first observe that the performance of the sum of states estimator is inconsistent for different basis sizes, and even for a given basis, it strongly depends on the chosen value of the cut-off. On the other hand, the determinant state estimator displays consistent improvement as the basis size increases, and systematically both outperforms the sum of states estimator and reaches the optimal accuracy within the subspace, further confirming that it is the most adapted estimator of the Rayleigh matrix in the context of Bridge.

Just like in the previously mentioned $h = 2h_c$ case, the dominating source of error is discretization error, meaning that Bridge is expected to significantly improve performance. Indeed, when using the determinant state estimator, we see improvements in the fidelity of the dynamics for all values of $\delta$, ranging from 2 to 5 orders of magnitude. As $\delta$ increases, we expect that the quality of the states degrade to the extent that Bridge cannot improve them as much as for lower values of $\delta$.

\begin{figure*}[t]
\centering
\includegraphics[width=\linewidth]{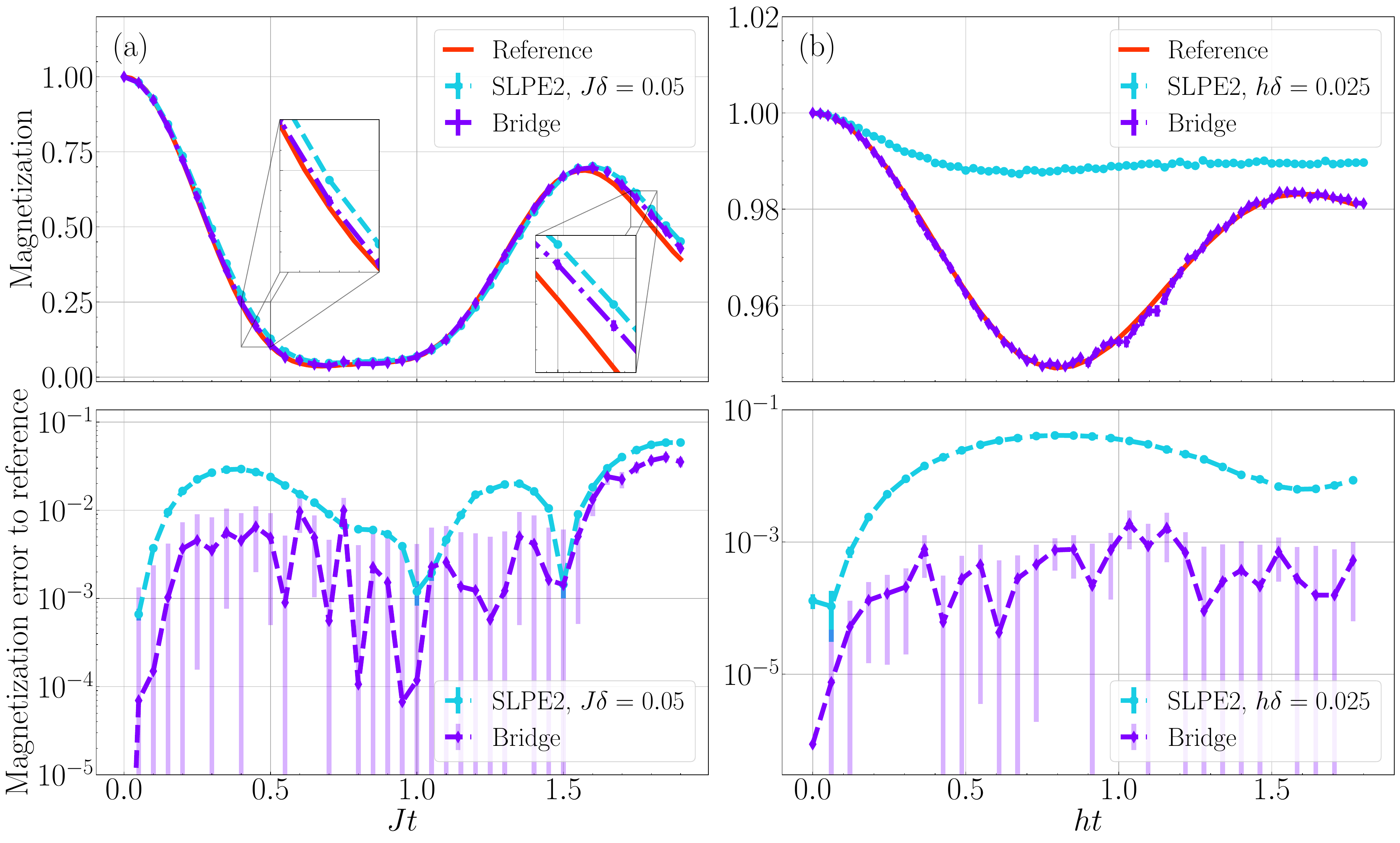}
\caption{Magnetization on the $8 \times 8$ transverse-field Ising model with $h=0.1h_c$ (a) and $h=2h_c$ (b). The original states are vision transformers optimized with p-tVMC using the scheme SLPE2 with time step $J\delta = 0.05$ (a) and $h\delta = 0.025$ (b). The reference for (a) are vision transformers optimized with p-tVMC using the higher order scheme SLPE3 with a smaller time step $J\delta = 0.025$. The reference for (b) is convolutional neural network states obtained with t-VMC in Ref.~\cite{Schmitt2020MarkusMarkus}, since obtaining good p-tVMC states in this regime would have been too computationally expensive.}
\label{fig:8x8}
\end{figure*}

\paragraph{}
In \cref{fig:8x8}, we present a practical application of Bridge to a system size of $8\times 8$ lattice sites, beyond the capabilities of exact calculations.
Similarly to the results on the $4 \times 4$ lattice, we see the most improvement in the case $h=2h_c$, where we expect to be dominated by discretization error. In this case, the numerical results show that Bridge faithfully reproduces the behaviour of the magnetization, even though the basis states do not.
This demonstrates the usefulness of Bridge as a practical method that can improve dynamics at a negligible cost.
For that practical application, obtaining the p-tVMC states takes about 12 hours on 4 H100 GPUs for (a) and 17 hours on 12 H100 GPUs for (b), while running Bridge on a single A100 GPU takes 15 minutes for (a) and 60 minutes for (b).

\paragraph{Restarting dynamics from a Bridge state}
In subplot (b) of \cref{fig:8x8}, going beyond times $ht = 2$ reveals a breakdown of Bridge. The interpretation is that after some time, the original basis states have gone so far away from the exact dynamics that we cannot expect Bridge to be able to make use of them. However, since Bridge successfully improved the dynamics at earlier times, we could run p-tVMC again from time $ht = 2$, starting from the corresponding Bridge state. Doing so, however, requires compressing the linear combination state $\sum_k \alpha_k(t) \ket{\phi_{\theta_k}}$ back into a single variational state $\ket{\phi_\theta}$. This operation can be performed with infidelity optimization \cite{Gravina2024PTVMC}, but it has revealed itself especially challenging in this case. Further investigation into the compression of linear states could allow restarting the dynamics from Bridge solutions whenever the original dynamics strays too far from the exact dynamics.

\section{Conclusion}
In this paper, we presented the determinant state mapping, which is a natural formalism to treat subspaces of the Hilbert space.
This mapping results in both a natural notion of distance between subspaces, and a natural notion of energy which is consistent with the one suggested by subspace methods. 
The resulting formalism is applicable in the context of variational states, and is compatible with the tools developed for Monte Carlo estimation, while providing new ones. 
It recontextualizes existing methods and provides a framework where new methods can be developed.

We then examined Bridge, a computationally cheap subspace method relying on tools from this framework, which can improve the accuracy of variational dynamics obtained by other established methods.
We discussed different estimators for the Rayleigh matrix, and showed that the determinant state estimator is the most suitable estimator in the context of Bridge.

Numerical evidence suggests that Bridge can significantly reduce time step discretization errors, while it does not provide significant improvements for optimization (or variationally-induced) errors. 
We saw that it is a powerful interpolation tool, while its extrapolation capabilities are more limited.

The low cost and potentially significant improvements provided by Bridge thus make it appear as a method that can be systematically used to post-process states approximating dynamics.

An important point that should be further examined is the compression of linear states, in order to restart the dynamics from Bridge states and further improve the performance of both the original time evolution method and Bridge.
Further investigations into Bridge could also examine the impact of enriching the subspace with more states, for examples Krylov states built from any of the $\ket{\phi_k}$.

More importantly, Bridge is thus evidence of the usefulness of subspace methods in the context of variational methods, especially when using the formalism of the determinant state mapping. It therefore invites more investigation into the potential new methods and tools that could be derived from using the determinant state mapping.

\begin{acknowledgments}
Simulations were performed with NetKet \cite{Netket2, Netket3}, and at times parallelized with mpi4JAX \cite{MPI4JAX}.
This software is built on top of JAX \cite{Jax2018Github}, equinox \cite{Kidger2021Equinox} and Flax \cite{Flax2020Github}.
We used QuTiP \cite{QuTiP1, QuTiP2} for exact benchmarks and mpmath for arbitrary precision arithmetics \cite{Mpmath}.
p-tVMC simulations where performed using the repository from~\cite{Gravina2024PTVMC}.
The accompanying code for the manuscript is available at \href{https://github.com/NeuralQXLab/variational-subspace-methods}{github.com/NeuralQXLab/variational-subspace-methods}.

We acknowledge extensive discussions with A. Shokry and R. Nutakki, as well as insightful discussions with G. Booth, G. Carleo, Z. Denis, A. Sinibaldi and V. Savona.
F.V. and A.K. acknowledge support by the French Agence Nationale de la Recherche through the NDQM project, grant ANR-23-CE30-0018.
This project was provided with computing HPC and storage resources by GENCI at IDRIS thanks to the grants 2023-AD010514908 and 2024-A0170515698 on the supercomputer Jean Zay's V100/A100 partition.

\textit{We also acknowledge a concurrent work by D. Hendry, A. Sinibaldi, and G. Carleo \cite{Hendry2025Grassmann}, where a similar formalism is discussed to describe subspaces and used to estimate excited states. This work appeared simultaneously on the preprint server.}
\end{acknowledgments}

\bibliographystyle{quantum}
\bibliography{bibliography}

\newpage

\appendix

\onecolumngrid

\section{Determinant state formalism}
\label{app:antisymmetric_states}

\subsection{Antisymmetric subspace and antisymmetrizer}

We recall the definitions of the antisymmetric subspace and the antisymmetrizer operator.

\begin{definition}[Permutation operator]
\label{def:permutation_operator}
The permutation operator $\bm{P_\sigma}$ over the Hilbert space $\mathcal H^{\otimes m}$ corresponding to a permutation $\sigma \in S_m$ is defined on product states by
\begin{equation}
\bm{P_\sigma} \ket{\phi_1} \otimes \ket{\phi_2} \otimes \ldots \otimes \ket{\phi_m} = \ket{\phi_{\sigma^{-1}(1)}} \otimes \ket{\phi_{\sigma^{-1}(2)}} \otimes \ldots \otimes \ket{\phi_{\sigma^{-1}(m)}}.
\end{equation}
\end{definition}
The reason why the states are permuted with $\sigma^{-1}$ instead of $\sigma$ is because only this definition produces a representation of $S_m$ on $\mathcal H^{\otimes m}$, that is, it is the only one satisfying $\bm{P_{\sigma \sigma'}} = \bm{P_\sigma} \bm{P_{\sigma'}}$.

\begin{definition}[Antisymmetric subspace]
\label{def:antisymmetric_subspace}
The antisymmetric subspace $V_-(\mathcal H^{\otimes m})$ is defined by
\begin{equation}
V_-(\mathcal H^{\otimes m}) = \left\{ \bm{\ket{\psi}} \in \mathcal H^{\otimes m} 
\text{ such that } \forall \sigma \in S_m, \ \bm{P_\sigma} \bm{\ket{\psi}} = \sgn(\sigma) \bm{\ket{\psi}} \right\}.
\end{equation}
\end{definition}

\begin{definition}[Antisymmetrizer operator]
\label{def:antisymmetrizer}
The antisymmetrizer $\bm{S_-}$ is the operator of $\mathcal H^{\otimes m}$ defined by
\begin{equation}
\bm{S_-} = \frac{1}{m!} \sum_\sigma \sgn (\sigma) \bm{P_\sigma},
\end{equation}
where $\bm{P_\sigma}$ is the permutation operator that permutes the states according to the permutation $\sigma$.
\end{definition}

We also recall several simple properties of the antisymmetrizer that will be useful throughout the appendix.
\begin{proposition}
\label{prop:antisymmetrizer_projector}
The antisymmetrizer $\bm{S_-}$ is the orthogonal projector on the antisymmetric subspace $V_-(\mathcal H^{\otimes m})$. In other words, it satisfies:
\begin{itemize}
\item $\bm{S_-}^2 = \bm{S_-}$,
\item $\bm{S_-}^\dag = \bm{S_-}$,
\item $\text{im} \left( \bm{S_-} \right) = V_-(\mathcal H^{\otimes m})$.
\end{itemize}
\end{proposition}

\begin{proposition}
\label{prop:antisymmetrizer_property}
The antisymmetrizer satisfies
\begin{equation}
\forall \sigma, \ \bm{S_-} \bm{P_\sigma} = \bm{P_\sigma} \bm{S_-} = \sgn(\sigma) \bm{S_-}.
\end{equation}
\end{proposition}

Finally, before getting to prove proper results, we show a simple equality that we will extensively use throughout this section.

\begin{proposition}
\label{prop:det_slater}
Let $\left\{ \ket{\phi_k} \right\}_k$ and $\left\{ \ket{\varphi_k} \right\}_k$ be two families of $m$ states of $\mathcal H$. If we define by $\bm{\ket \phi}$ and $\bm{\ket \varphi}$ the product states of each family
\begin{align}
\bm{\ket \phi} &= \ket{\phi_1} \otimes \ket{\phi_2} \otimes \ldots \otimes \ket{\phi_m}, \\
\bm{\ket \varphi} &= \ket{\varphi_1} \otimes \ket{\varphi_2} \otimes \ldots \otimes \ket{\varphi_m},
\end{align}
and by $\bm {\ket{\phi_A} = S_- \ket{\phi}}$ and $\bm {\ket{\varphi_A} = S_- \ket{\varphi}}$ their antisymmetrization, then we have
\begin{equation}
\bm{\braket{\phi_A | \varphi_A}} = \bm{\braket{\phi | S_- | \varphi}} = \frac{1}{m!} \det S,
\end{equation}
where $S$ is the $m \times m$ overlap matrix defined by $S_{ij} = \braket{\phi_i | \varphi_j}$.
\end{proposition}

\begin{proof}
We have
\begin{align}
\bm{\braket{\phi_A | \varphi_A}} &= \bm{\braket{\phi | S_-^\dag S_- | \varphi}} \\
&= \bm{\braket{\phi | S_- | \varphi}},
\end{align}
where we used \cref{prop:antisymmetrizer_projector}. We now expand the antisymmetrizer
\begin{align}
\bm{\braket{\phi | S_- | \varphi}} &= \frac{1}{m!} \sum_\sigma \sgn(\sigma) \braket{\phi_1 | \varphi_{\sigma(1)}} \braket{\phi_2 | \varphi_{\sigma(2)}} \ldots \braket{\phi_m | \varphi_{\sigma(m)}} \\
&= \frac{1}{m!} \det S.
\end{align}
\end{proof}

\subsection{Determinant state mapping}

We now prove the property that justifies the definition of the determinant state mapping (\cref{eq:determinant_state_mapping}).

\begin{proposition}[Slater determinant under a change of basis]
\label{prop:slater_cob}
Let $\left\{ \ket{\phi_k} \right\}_k$ be a family of $m$ states of $\mathcal H$, and let $\left\{ \ket{\varphi_k} \right\}_k$ be the family of $m$ states defined by the linear transformation
\begin{equation}
\ket{\varphi_p} = \sum_{k=1}^m B_{kp} \ket{\phi_k},
\end{equation}
where $B$ is an invertible $m \times m$ matrix. Then, the antisymmetrization of both product states are colinear, and we have
\begin{equation}
\bm{S_-} \ket{\varphi_1} \otimes \ldots \otimes \ket{\varphi_m} = \det(B) \bm{S_-} \ket{\phi_1} \otimes \ldots \otimes \ket{\phi_m}.
\end{equation}
\end{proposition}

\begin{proof}
We start by expanding the $\ket{\varphi_p}$,
\begin{equation}
\label{eq:expansion_in_cob}
\bm{S_-} \ket{\varphi_1} \otimes \ldots \otimes \ket{\varphi_m} = \sum_{k_1, \ldots, k_m} B_{k_1, 1} \ldots B_{k_m, m} \bm{S_-} \ket{\phi_{k_1}} \otimes \ldots \otimes \ket{\phi_{k_m}}.
\end{equation}

From \cref{prop:det_slater}, the amplitude of the antisymmetrized state $\bm{S_-} \ket{\phi_{k_1}} \otimes \ldots \otimes \ket{\phi_{k_m}}$ corresponding the basis element $\bm s = (s_1, \ldots, s_m)$ is given by $\det \Phi (\bm s)$, where $\left[ \Phi (\bm s) \right]_{ij} = \braket{s_i | \phi_{k_j}}$. Therefore, if the $m$-uplet $(k_1, \ldots, k_m)$ has at least one repeated value, then $\Phi (\bm s)$ has repeated columns. As such, for any basis state $\bm s$, we have $\det \Phi (\bm s) = 0$, and so we conclude that $\bm{S_-} \ket{\phi_{k_1}} \otimes \ldots \otimes \ket{\phi_{k_m}} = 0$.

We can thus restrict the sum in \cref{eq:expansion_in_cob} to only $m$-uplet with no repeated values. An $m$-uplet of values $k_i \in \{ 1, \ldots, m \}$ such that no two $k_i$ are identical is effectively a permutation $\sigma : i \mapsto k_i$. We can thus rewrite this expression as a sum over permutations,
\begin{align}
\bm{S_-} \ket{\varphi_1} \otimes \ldots \otimes \ket{\varphi_m} &= \sum_\sigma B_{\sigma(1), 1} \ldots B_{\sigma(m), m} \bm{S_-} \ket{\phi_{\sigma(1)}} \otimes \ldots \otimes \ket{\phi_{\sigma(m)}} \\
&= \sum_\sigma B_{\sigma(1), 1} \ldots B_{\sigma(m), m} \bm{S_-} \bm{P_{\sigma^{-1}}} \ket{\phi_1} \otimes \ldots \otimes \ket{\phi_m} \\
&= \left( \sum_\sigma \sgn(\sigma) B_{\sigma(1), 1} \ldots B_{\sigma(m), m} \right) \bm{S_-} \ket{\phi_1} \otimes \ldots \otimes \ket{\phi_m} \\
&= \det(B) \bm{S_-} \ket{\phi_1} \otimes \ldots \otimes \ket{\phi_m}.
\end{align}
where we used \cref{prop:antisymmetrizer_property}.
\end{proof}

We have now proven \cref{eq:det_cob}, which ensures that the determinant state mapping \cref{eq:determinant_state_mapping} is well-defined.

We now show that the determinant state mapping is also injective. We first prove a lemma that is also useful in interpreting the distance between determinant states defined by \cref{eq:subspace_fidelity}.

\begin{lemma}[Fidelity between determinant states]
\label{lem:det_state_fidelity}
Let $U$ and $V$ be $m$-dimensional subspaces of $\mathcal H$ with respective \textbf{orthonormal} bases $\left\{ \ket{\psi_k} \right\}_k$ and $\left\{ \ket{\phi_k} \right\}_k$. 
Then, the fidelity between the corresponding determinant states can be expressed as
\begin{equation}
\frac{\left| \bm{\braket{\psi_A | \phi_A}} \right|^2}{\bm{\lVert \bm{\psi_A} \rVert}^2 \bm{\lVert \phi_A \rVert}^2} = \det G_{\tilde \phi},
\end{equation}
where $G_{\tilde \phi}$ is the Gram matrix of the image of the family $\left\{ \ket{\phi_k} \right\}_k$ by $P_U$ the orthogonal projector on $U$,
\begin{equation}
\left[ G_{\tilde \phi} \right]_{ij} = \braket{\phi_i | P_U^\dag P_U | \phi_j}.
\end{equation}
\end{lemma}

\begin{proof}
First, we emphasize that the left-hand term does not depend on the chosen bases. It is only the right-hand term that requires $G_{\tilde \phi}$ to be the Gram matrix of the projection of an orthonormal basis of $V$.

Using \cref{prop:det_slater}, we can rewrite the left-hand term as
\begin{equation}
\frac{\left| \bm{\braket{\psi_A | \phi_A}} \right|^2}{\bm{\lVert \bm{\psi_A} \rVert}^2 \bm{\lVert \phi_A \rVert}^2} = \frac{\left| \det S \right|^2}{\det G_\psi \det G_\phi},
\end{equation}
where the matrix $S$ is defined by $S_{ij} = \braket{\psi_i | \phi_j}$ and the matrices $G_\psi$ and $G_\phi$ are the Gram matrices of the families $\left\{ \ket{\psi_k} \right\}_k$ and $\left\{ \ket{\phi_k} \right\}_k$ respectively. Since we have chosen the two families to be orthonormal, it follows that these two Gram matrices are the identity.

Thus, only the numerator is left. It can be rewritten as
\begin{align}
\frac{\left| \bm{\braket{\psi_A | \phi_A}} \right|^2}{\bm{\lVert \bm{\psi_A} \rVert}^2 \bm{\lVert \phi_A \rVert}^2} &= \left| \det S \right|^2 \\
&= (\det S)^* (\det S) \\
&= \det \left( S^\dag S \right).
\end{align}
The matrix elements of $S^\dag S$ are given by
\begin{equation}
\left[ S^\dag S \right]_{ij} = \sum_k \braket{\phi_i | \psi_k} \braket{\psi_k | \phi_j}.
\end{equation}
We identify $P_U = \sum_k \ket{\psi_k} \bra{\psi_k}$ the orthogonal projector on $U$, such that we can rewrite the elements of $S^\dag S$ as
\begin{align}
\left[ S^\dag S \right]_{ij} &= \braket{\phi_i | P_U | \phi_j} \\
&= \braket{\phi_i | P_U^2 | \phi_j} \\
&= \braket{\phi_i | P_U^\dag P_U | \phi_j} \\
&= \braket{\tilde \phi_i |\tilde \phi_j}.
\end{align}
where we introduced $\ket{\tilde \phi_k} = P_U \ket{\phi_k}$. We also introduce $G_{\tilde \phi}$, the Gram matrix of the projected family $\left\{ \ket{\tilde \phi_k} \right\}_k$, defined by
\begin{equation}
\left[ G_{\tilde \phi} \right]_{ij} = \braket{\tilde \phi_i |\tilde \phi_j},
\end{equation}
such that we have $S^\dag S = G_{\tilde \phi}$. We indeed conclude that
\begin{equation}
\frac{\left| \bm{\braket{\psi_A | \phi_A}} \right|^2}{\bm{\lVert \bm{\psi_A} \rVert}^2 \bm{\lVert \phi_A \rVert}^2} = \det G_{\tilde \phi}.
\end{equation}
\end{proof}

This lemma allow us to easily prove the injectivity of the determinant state mapping.

\begin{proposition}[Injectivity of the determinant state mapping]
\label{prop:injectivity}
The mapping
\begin{equation}
\begin{cases}
\gr_m(\mathcal H) &\to S / (\mathbb C, \times) \\
V &\mapsto \left[ S_- \ket{\phi_1} \otimes \ldots \otimes \ket{\phi_m} \right]
\end{cases}
\end{equation}
is injective.
\end{proposition}

\begin{proof}
Let $U$ and $V$ be $m$-dimensional subspaces of $\mathcal H$ with respective bases $\left\{ \ket{\psi_k} \right\}_k$ and $\left\{ \ket{\phi_k} \right\}_k$, such that $\bm{\ket{\psi_A}}$ and $\bm{\ket{\phi_A}}$ are the same quantum state. We want to show that $U = V$.

Since $\bm{\ket{\psi_A}}$ and $\bm{\ket{\phi_A}}$ are the same quantum state, their fidelity is 1,
\begin{equation}
\frac{\left| \bm{\braket{\psi_A | \phi_A}} \right|^2}{\bm{\lVert \bm{\psi_A} \rVert}^2 \bm{\lVert \phi_A \rVert}^2} = 1.
\end{equation}
Without loss of generality, $\left\{ \ket{\psi_k} \right\}_k$ and $\left\{ \ket{\phi_k} \right\}_k$ can be chosen to be orthonormal bases. From \cref{lem:det_state_fidelity}, it follows that the Gram matrix $G_{\tilde \phi}$ has determinant 1. Using the Hadamard inequality on the Gram determinant, we deduce that
\begin{equation}
\prod_k \lVert P_U \ket{\phi_k} \rVert^2 \geq 1.
\end{equation}
However, since the vectors $\ket{\phi_k}$ are normalized, we also have
\begin{equation}
\forall k, \ \lVert P_U \ket{\phi_k} \rVert^2 \leq 1.
\end{equation}
From both inequalities, we deduce that
\begin{equation}
\forall k, \ \lVert P_U \ket{\phi_k} \rVert^2 = 1.
\end{equation}
This implies that each state $\ket{\phi_k}$ is in $U$. Since they form a basis of $V$, it follows that $V \subset U$, and since they have the same dimension, we conclude that $U = V$. We have thus shown that if two subspaces have the same determinant state, then they are the same subspace. In other words, the determinant state mapping is injective.
\end{proof}

\subsection{Distance between subspaces}

In \cref{sec:subspace_fidelity}, we introduced the distance between subspaces $U$ and $V$ with respective bases $\left\{ \ket{\psi_k} \right\}_k$ and $\left\{ \ket{\phi_k} \right\}_k$ defined by
\begin{equation}
d_s \left( U, V \right) = d \left( \bm{\ket{\psi_A}}, \bm{\ket{\phi_A}} \right).
\end{equation}
To help better understand this distance, the next proposition characterizes maximally distant subspaces under that metric.

\begin{proposition}[Maximally distant subspaces]
\label{prop:max_distance_subspaces}
Two $m$-dimensional subspaces $U$ and $V$ of $\mathcal H$ have the maximum distance of $\pi/2$ for $d_s$, if and only if $V$ contains a non-zero vector orthogonal to $U$.
\end{proposition}

\begin{proof}
We use the characterization of $d_s$ given by \cref{lem:det_state_fidelity}. If $V$ contains a non-zero vector $\ket{\phi_1}$ orthogonal to $U$, then this vector can be completed into an orthogonal basis $\left\{ \ket{\phi_k} \right\}_k$ of $V$. The orthogonal projection of this family onto $U$ is not full rank since the first vector is $0$, and the resulting Gram matrix $G_{\tilde \phi}$ is thus singular. The fidelity between the subspaces is 0, and so their distance is $\pi/2$.

Conversely, suppose that $d_s(U, V) = \pi/2$. Again, using \cref{lem:det_state_fidelity}, this implies that for any orthogonal basis $\left\{ \ket{\phi_k} \right\}_k$ of $V$, the Gram matrix $G_{\tilde \phi}$ of the orthogonal projection of these states on $U$ satisfies
\begin{equation}
\det G_{\tilde \phi} = \cos d_s(U, V) = 0.
\end{equation}
We conclude that the states $P_U \ket{\phi_k}$ are linearly dependent and that there exists $\alpha \in \mathbb C^m, \ \alpha \neq 0$ such that
\begin{align}
&\sum_k \alpha_k P_U \ket{\phi_k} = 0 \\
\iff &P_U \sum_k \alpha_k \ket{\phi_k} = 0.
\end{align}
Since the states $\ket{\phi_k}$ are linearly independent, we conclude that $\sum_k \alpha_k \ket{\phi_k}$ is a non-zero state of $V$ that is orthogonal to $U$.
\end{proof}

\subsection{Energy of a subspace}

\begin{proposition}[Energy of the determinant state \cite{Pfau2024}]
\label{prop:det_energy}
Let $H$ be a Hamiltonian acting on $\mathcal H$, and let $V$ be an $m$-dimensional subspace of $\mathcal H$ spanned by a basis $\left\{ \ket{\phi_k} \right\}_k$. The energy of the subspace defined through the determinant state is the trace of its Rayleigh matrix,
\begin{equation}
\frac{\bm{\braket{\phi_A | H | \phi_A}}}{\bm{\braket{\phi_A | \phi_A}}} = \tr \left( G^{-1} G^{(H)} \right).
\end{equation}
where $G$ is the Gram matrix and $G^{(H)}$ is the projected Hamiltonian
\begin{align}
G_{ij} &= \braket{\phi_i | \phi_j}, \\
G_{ij}^{(H)} &= \braket{\phi_i | H | \phi_j}.
\end{align}
\end{proposition}

\begin{proof}
For convenience, we denote by $\bm{\ket{\phi}} = \ket{\phi_1} \otimes \ldots \otimes \ket{\phi_m}$ the product state formed from the basis $\left\{ \ket{\phi_k} \right\}_k$.

We start by rewriting the left-hand term using \cref{prop:antisymmetrizer_projector} and the fact that $\bm H$ commutes with $\bm{S_-}$,\footnote{The Hamiltonian $\bm H$ commutes with each of the permutation operators $\bm{P_\sigma}$, so it also commutes with $\bm{S_-}$.}
\begin{align}
\frac{\bm{\braket{\phi_A | H | \phi_A}}}{\bm{\braket{\phi_A | \phi_A}}} &= \frac{\bm{\braket{\phi | S_-^\dag H S_- | \phi}}}{\bm{\braket{\phi_A | \phi_A}}} \\
&= \frac{\bm{\braket{\phi | S_-^2 H | \phi}}}{\bm{\braket{\phi_A | \phi_A}}} \\
&= \frac{\bm{\braket{\phi | S_- H | \phi}}}{\bm{\braket{\phi_A | \phi_A}}} \\
&= \frac{\bm{\braket{\phi_A | H | \phi}}}{\bm{\braket{\phi_A | \phi_A}}} \\
&= \sum_p \frac{\bm{\braket{\phi_A | H_p | \phi}}}{\bm{\braket{\phi_A | \phi_A}}}.
\end{align}
\end{proof}

According to \cref{prop:det_slater}, the denominator can be expressed as $\bm{\braket{\phi_A | \phi_A}} = \frac{1}{m!} \det G$. 

We now consider the numerator $\bm{\braket{\phi_A | H_p | \phi}}$. If we define the family $\left\{ \ket{\tilde \phi_k} \right\}_k$ by
\begin{equation}
\ket{\tilde \phi_k} =
\begin{cases}
\ket{\phi_k} &\text{ if } k \neq p \\
H \ket{\phi_k} &\text{ if } k = p
\end{cases},
\end{equation}
then the numerator is
\begin{align}
\bm{\braket{\phi_A | H_p | \phi}} &= \bm{\braket{\phi_A | \tilde \phi}}, \\
&= \frac{1}{m!} \det G_{p \to p}^{(H)},
\end{align}
where we used \cref{prop:det_slater}, and we defined $G_{p \to p}^{(H)}$, the matrix $G$ where the $p$-th column was replaced by the $p$-th column of $G^{(H)}$,
\begin{equation}
\left[ G_{p \to p}^{(H)} \right]_{ij} = 
\begin{cases}
\braket{\phi_i | \phi_j} &\text{ if } j \neq p \\
\braket{\phi_i | H | \phi_p} &\text{ if } j = p
\end{cases}.
\end{equation}

We have thus obtained
\begin{equation}
\frac{\bm{\braket{\phi_A | H | \phi_A}}}{\bm{\braket{\phi_A | \phi_A}}} = \sum_p \frac{\det G_{p \to p}^{(H)}}{\det G}.
\end{equation}
As we pointed out, the matrix $G_{p \to p}^{(H)}$ is nothing but the matrix $G$ where the $p$-th column was replaced by the $p$-th column of $G^{(H)}$. We can therefore use Cramer's rule to rewrite the ratio of determinants. If we denote the $p$-th column of $G^{(H)}$ by $\left[ G^{(H)} \right]_{:,p}$, we can rewrite the energy as
\begin{align}
\frac{\bm{\braket{\phi_A | H | \phi_A}}}{\bm{\braket{\phi_A | \phi_A}}} &= \sum_p \left[ G^{-1} \left[ G^{(H)} \right]_{:,p} \right]_p \\
&= \sum_p \left[ G^{-1} G^{(H)} \right]_{pp} \\
&= \tr \left( G^{-1} G^{(H)} \right).
\end{align}

\begin{proposition}[Estimator of the Rayleigh matrix \cite{Pfau2024}]
\label{prop:rm_estimator}
The Rayleigh matrix can be decomposed as
\begin{equation}
G^{-1} G^{(H)} = \mathbb{E}_{\bm s \sim |\bm{\phi_A(s)}|^2} \left[ \Phi(\bm s)^{-1} \Phi^{(H)}(\bm s) \right],
\end{equation}
with
\begin{align}
\left[ \Phi(\bm s) \right]_{ij} &= \braket{s_i | \phi_j}, \\
\left[ \Phi^{(H)}(\bm s) \right]_{ij} &= \braket{s_i | H | \phi_j}.
\end{align}
\end{proposition}

\begin{proof}
We consider the matrix element $(k,p)$ of the Rayleigh matrix. Using Cramer's rule again, we can reexpress it as a ratio of determinants,
\begin{align}
\left[ G^{-1} G^{(H)} \right]_{kp} &= \left[ G^{-1} G^{(H)}_{:,p} \right]_k \\
&= \frac{\det G_{k \to p}^{(H)}}{\det G}.
\end{align}
where $G_{k \to p}^{(H)}$, the matrix $G$ where the $k$-th column was replaced by the $p$-th column of $G^{(H)}$.

If we define the family $\left\{ \ket{\tilde \phi_i} \right\}_i$ by
\begin{equation}
\ket{\tilde \phi_i} =
\begin{cases}
\ket{\phi_i} &\text{ if } i \neq k \\
H \ket{\phi_p} &\text{ if } i = k
\end{cases},
\end{equation}
then the matrix $G_{k \to p}^{(H)}$ can be written as $\left[ G_{k \to p}^{(H)} \right]_{ij} = \braket{\phi_i | \tilde \phi_j}$. Using \cref{prop:det_slater}, we can therefore rewrite the matrix element as
\begin{equation}
\left[ G^{-1} G^{(H)} \right]_{kp} = \frac{\bm{\braket{\phi_A | \tilde \phi_A}}}{\bm{\braket{\phi_A | \phi_A}}}.
\end{equation}
From there, we can follow the usual steps to provide a decomposition as an expectation value,
\begin{align}
\left[ G^{-1} G^{(H)} \right]_{kp} &= \sum_{\bm s} \frac{\bm{\braket{\phi_A | s}} \bm{\braket{s | \tilde \phi_A}}}{\bm{\braket{\phi_A | \phi_A}}} \\
&= \sum_{\substack{\bm s \\ \bm{\braket{\phi_A | s}} \neq 0}} \frac{\bm{\braket{\phi_A | s}}}{\bm{\braket{\phi_A | \phi_A}}} \bm{\braket{s | \tilde \phi_A}} \\
&= \sum_{\substack{\bm s \\ \bm{\braket{\phi_A | s}} \neq 0}} \frac{\left| \bm{\braket{\phi_A | s}} \right|^2}{\bm{\braket{\phi_A | \phi_A}}} \frac{\bm{\braket{s | \tilde \phi_A}}}{\bm{\braket{s | \phi_A}}} \\
&= \label{eq:expectation_value}\mathbb{E}_{\bm s \sim |\bm{\phi_A(s)}|^2} \left[ \frac{\bm{\braket{s | \tilde \phi_A}}}{\bm{\braket{s | \phi_A}}} \right].
\end{align}

Let us now examine the object inside the expectation value. Applying \cref{prop:det_slater} yet again, we can reexpress both the numerator and denominator as determinants,
\begin{align}
\bm{\braket{s | \tilde \phi_A}} &= \frac{1}{m!} \det \Phi_{k \to p}^{(H)}(\bm s), \\
\bm{\braket{s | \phi_A}} &= \frac{1}{m!} \det \Phi(\bm s),
\end{align}
where we defined $\left[ \Phi(\bm s) \right]_{ij} = \braket{s_i | \phi_j}$ and
\begin{align}
\left[ \Phi_{k \to p}^{(H)}(\bm s) \right]_{ij} &= \braket{s_i | \tilde \phi_j} \\
&=
\begin{cases}
\braket{s_i | \phi_j} &\text{ if } j \neq k \\
\braket{s_i | H | \phi_p} &\text{ if } j = k
\end{cases}.
\end{align}
In other words, if we define the matrix $\Phi^{(H)}(\bm s)$ by $\left[ \Phi^{(H)}(\bm s) \right]_{ij} = \braket{s_i | H | \phi_j}$, then the matrix $\Phi_{k \to p}^{(H)}(\bm s)$ is $\Phi(\bm s)$ where the $k$-th column was replaced by the $p$-th column of $\Phi^{(H)}(\bm s)$. Using Cramer's rule again, we can therefore reexpress
\begin{align}
\frac{\bm{\braket{s | \tilde \phi_A}}}{\bm{\braket{s | \phi_A}}} &= \frac{\det \Phi_{k \to p}^{(H)}(\bm s)}{\det \Phi(\bm s)} \\
&= \left[ \Phi(\bm s)^{-1} \Phi^{(H)}(\bm s) \right]_{kp}.
\end{align}
Inserting this expression in \cref{eq:expectation_value} for every matrix element, we obtain the final result,
\begin{equation}
G^{-1} G^{(H)} = \mathbb{E}_{\bm s \sim |\bm{\phi_A(s)}|^2} \left[ \Phi(\bm s)^{-1} \Phi^{(H)}(\bm s) \right].
\end{equation}
\end{proof}

\begin{proposition}[Zero-variance property of the Rayleigh matrix estimator]
\label{prop:zero-variance}
If the states $\ket{\phi_k}$ span an invariant subspace of $H$, then the local energy matrix $\Phi(\bm s)^{-1} \Phi^{(H)}(\bm s)$ is a constant function of $\bm s$.
\end{proposition}

\begin{proof}
If the states $\ket{\phi_k}$ span an invariant subspace, then there exists an $(m,m)$ matrix $B$ such that
\begin{equation}
\forall p, \ H \ket{\phi_p} = \sum_k B_{kp} \ket{\phi_k}.
\end{equation}
The matrix $\Phi^{(H)}(\bm s)$ can therefore be rewritten as
\begin{align}
\left[ \Phi^{(H)}(\bm s) \right]_{ij} &= \braket{s_i | H | \phi_j} \\
&= \sum_k B_{kj} \braket{s_i | \phi_k} \\
&= \left[ \Phi(\bm s) B \right]_{ij}.
\end{align}
We can then insert this expression in the Rayleigh matrix to obtain
\begin{equation}
\Phi(\bm s)^{-1} \Phi^{(H)}(\bm s) = B,
\end{equation}
such that the local energy matrix is independent of the value of the sample $\bm s$. We conclude that the variance of the estimator of the Rayleigh matrix induced by \cref{prop:rm_estimator} is therefore 0 if the states $\ket{\phi_k}$ span an invariant subspace. 
\end{proof}

\section{Eigenvalues of the Rayleigh matrix as the generalization of eigenvalues in a subspace}
\label{app:ritz_values}

In the main text, we introduced the eigenvalues $\mu_k$ of the Rayleigh matrix $G^{-1} G^{(H)}$. In this section, we justify why the $\mu_k$, called the \textit{Ritz values}, can be understood as the natural generalization of the eigenvalues of an operator when restricted to a subspace that is not necessarily an invariant subspace. In particular, we will prove the generalized variational principle of \cref{eq:cit}.

We start by recalling two important results of linear algebra: the min-max theorem and the Cauchy interlacing theorem. While these theorems are expressed in the context of an orthonormal basis, we will see that their results can be easily extended to an arbitrary basis.

\begin{theorem}[Min-max theorem \cite{Golub}]
\label{th:minmax}
Let $H$ be an $n \times n$ Hermitian operator acting on the Hilbert space $\mathcal H$, and let $E_k, \ k \in \left\{ 0, \ldots, n-1 \right\}$ be its eigenvalues indexed in increasing order. The eigenvalue $E_k$ can be expressed as a maximum over subspaces of $\mathcal H$ of dimension $k+1$,
\begin{equation}
E_k = \min_{W \in \gr_{k+1}(\mathcal H)} \max_{\substack{\ket{\psi} \in W \\ \ket{\psi} \neq 0}} \frac{\braket{\psi | H | \psi}}{\braket{\psi | \psi}},
\end{equation}
where we recall that $\gr_{k+1}(\mathcal H)$ is the set of subspaces of $\mathcal H$ with dimension $k+1$.
\end{theorem}

\begin{proof}
For simplicity, we will denote
\begin{equation}
E_{\max}(W) = \max_{\substack{\ket{\psi} \in W \\ \ket{\psi} \neq 0}} \frac{\braket{\psi | H | \psi}}{\braket{\psi | \psi}}.
\end{equation}
In order to prove that $E_k = \min_{W \in \gr_{k+1}(\mathcal H)} E_{\max}(W)$, we need to show:
\begin{align}
\textbf{(i)} \quad& \forall W \in \gr_{k+1}(\mathcal H), E_{\max}(W) \geq E_k \\
\textbf{(ii)} \quad& \exists W \in \gr_{k+1}(\mathcal H), \ E_{\max}(W) = E_k.
\end{align}

We denote by $\ket{E_k}$ the $k$-th eigenvector of $H$, corresponding to the eigenvalue $E_k$.

\paragraph{(i)}
Consider the $(n-k)$-dimensional subspace $V = \operatorname{span} \left( \ket{E_k}, \ket{E_{k+1}}, \ldots, \ket{E_{n-1}} \right)$. For any state $\ket{\psi_\alpha} = \sum_{p=k}^{n-1} \alpha_p \ket{E_p} \in V$, the Rayleigh quotient satisfies
\begin{equation}
\label{eq:wbar}
\frac{\braket{\psi_\alpha | H | \psi_\alpha}}{\braket{\psi_\alpha | \psi_\alpha}} = \frac{\sum_{p=k}^{n-1} E_p | \alpha_p |^2}{\sum_{p=k}^{n-1} | \alpha_p |^2} \geq E_k.
\end{equation}
Since any subspace $W \in \gr_{k+1}(\mathcal H)$ has dimension $k+1$, its intersection with the $(n-k)$-dimensional subspace $V$ is of dimension greater or equal to 1. Therefore, there exists a non-zero vector of $W$ that also belongs to $V$. By \cref{eq:wbar}, we deduce that $E_{\max}(W) \leq E_k$.

\paragraph{(ii)}
We consider the subspace $W^* = \operatorname{span} \left( \ket{E_0}, \ket{E_1}, \ldots, \ket{E_k} \right)$. For a state $\ket{\psi_\alpha} = \sum_{p=0}^{k} \alpha_p \ket{E_p} \in W^*$, the Rayleigh quotient is given by
\begin{equation}
\frac{\braket{\psi_\alpha | H | \psi_\alpha}}{\braket{\psi_\alpha | \psi_\alpha}} = \frac{\sum_{p=0}^{k} E_p | \alpha_p |^2}{\sum_{p=0}^{k} | \alpha_p |^2}.
\end{equation}
Its maximum is thus $E_{\max}(W^*) = E_k$.

\end{proof}

Based on the characterization of eigenvalues given by the min-max theorem, the following theorem justifies how the Ritz values $\mu_k$ are the natural extension of eigenvalues when restricted to a subspace. From a variational perspective, they are the smallest values that can be obtained within $V$.

\begin{theorem}[Cauchy interlacing theorem]
\label{th:cauchy_interlacing}
Let $H$ be an $n \times n$ Hermitian operator. Let $V$ be an $m$-dimensional subspace with \textbf{orthonormal} basis $\left\{ \ket{\phi_k} \right\}_k$. We define the matrix $G^{(H)}$ by $G_{ij}^{(H)} = \braket{\phi_i | H | \phi_j}$. If we denote by $E_k$ and $\mu_k$ the respective eigenvalues of $H$ and $G^{(H)}$, indexed in increasing order, then we have
\begin{equation}
\mu_k = \min_{\substack{W \in \gr_{k+1} \left(\mathcal H \right) \\ W \subset V}} \max_{\substack{\ket{\psi} \in W \\ \ket{\psi} \neq 0}} \frac{\braket{\psi | H | \psi}}{\braket{\psi | \psi}},
\end{equation}
and
\begin{equation}
\forall k \leq m, \ E_k \leq \mu_k.
\end{equation}
\end{theorem}

\begin{proof}
Applying the min-max theorem to $\mu_k$, we get
\begin{equation}
\mu_k = \min_{W \in \gr_{k+1} \left( V \right)} \max_{\substack{\alpha \in W \\ \alpha \neq 0}} \frac{\alpha^\dag G^{(H)} \alpha}{\alpha^\dag \alpha}.
\end{equation}

We can rewrite
\begin{align}
\alpha^\dag G^{(H)} \alpha &= \sum_{kp} \alpha_k^* \braket{\phi_k | H | \phi_p} \alpha_p \\
&= \left( \sum_{k} \alpha_k^* \bra{\phi_k} \right) H \left( \sum_p \alpha_p \ket{\phi_p} \right) \\
&= \braket{\psi_\alpha | H | \psi_\alpha},
\end{align}
where we defined $\ket{\psi_\alpha} = \sum_k \alpha_k \ket{\phi_k}$.

The denominator can also be reexpressed as
\begin{align}
\alpha^\dag \alpha &= \sum_{k} \alpha_k^* \alpha_k \\
&= \sum_{kp} \alpha_k^* \delta_{kp} \alpha_p \\
&= \sum_{kp} \alpha_k^* \braket{\phi_k | \phi_p} \alpha_p \\
&= \braket{\psi_\alpha | \psi_\alpha},
\end{align}
where we used the orthogonality of the family $\left\{ \ket{\phi_k} \right\}_k$.

We can thus rewrite the min-max theorem as
\begin{equation}
\mu_k = \min_{W \in \gr_{k+1} \left( V \right)} \max_{\substack{\alpha \in W \\ \alpha \neq 0}} \frac{\braket{\psi_\alpha | H | \psi_\alpha}}{\braket{\psi_\alpha | \psi_\alpha}}.
\end{equation}

The set $W_{\mathcal H} = \left\{ \ket{\psi_\alpha}, \alpha \in W \right\}$ defines a $(k+1)$-dimensional subspace of $\mathcal H$ which is included in $V$. Conversely, any $(k+1)$-dimensional subspace of $\mathcal H$ which is included in $V$ can be written as $W_{\mathcal H}$ for some subspace $W$ of $V$. We can therefore rewrite the min-max theorem on $\mu_k$ once again,
\begin{align}
\mu_k &= \min_{W \in \gr_{k+1} \left( V \right)} \max_{\substack{\ket{\psi} \in W_{\mathcal H} \\ \ket{\psi} \neq 0}} \frac{\braket{\psi | H | \psi}}{\braket{\psi | \psi}} \\
&= \min_{\substack{W \in \gr_{k+1} \left( \mathcal H \right) \\ W \subset V}} \max_{\substack{\ket{\psi} \in W \\ \ket{\psi} \neq 0}} \frac{\braket{\psi | H | \psi}}{\braket{\psi | \psi}}.
\end{align}
It then follows immediately from the min-max theorem on $E_k$ that
\begin{equation}
\forall k \leq m, \ E_k \leq \mu_k.
\end{equation}
\end{proof}

In the Cauchy interlacing theorem, the Ritz values $\mu_k$ are defined as the eigenvalues of the matrix $G^{(H)}$ constructed from an orthonormal basis of the subspace $V$. We note that for a given Hermitian operator $H$, they are independent of the chosen orthonormal basis, and only depend on the subspace $V$. The following result indicates how to obtain the Ritz values from matrices constructed from any basis of $V$.

\begin{proposition}[Ritz values from non-orthonormal families]
Let $H$ be an $n \times n$ Hermitian operator. Let $V$ be an $m$-dimensional subspace. The corresponding Ritz values $\mu_k$ can be obtained from any basis $\left\{ \ket{\phi_k} \right\}_k$ of $V$ as the eigenvalues of the Rayleigh matrix $G^{-1} G^{(H)}$, where $G_{ij} = \braket{\phi_i | H | \phi_j}$ and $G_{ij}^{(H)} = \braket{\phi_i | H | \phi_j}$.
\end{proposition}

\begin{proof}
Let $\left\{ \ket{\tilde \phi_k} \right\}_k$ be an orthogonalization of the family $\left\{ \ket{\phi_k} \right\}_k$ (for example through the Gram-Schmidt process), such that the matrices of columns of each families are related by $\tilde \Phi = \Phi A$ where $\Phi_{sk} = \braket{s | \phi_k}$ and $\tilde \Phi_{sk} = \braket{s | \tilde \phi_k}$, and $A$ is an invertible $m \times m$ matrix that defines the orthogonalization process.

The Ritz values $\mu_k$ are the eigenvalues of the matrix $\tilde G^{(H)} = \tilde \Phi^\dag H \tilde \Phi$. Expressed in terms of $\Phi$, we have $\tilde G^{(H)} = A^\dag \Phi^\dag H \Phi A = A^\dag G^{(H)} A$. The eigenvalues are unchanged upon conjugation by an invertible matrix such as $A$. Therefore, the $\mu_k$ are also the eigenvalues of $A \tilde G^{(H)} A^{-1} = A A^\dag G^{(H)}$. On the other hand, we have
\begin{align}
G &= \Phi^\dag \Phi \\
&= (A^{-1})^\dag \tilde \Phi^\dag \tilde \Phi A^{-1} \\
&= (A^{-1})^\dag A^{-1} \\
&= ( A A^\dag)^{-1}.
\end{align}
where we used the orthonormality of the family $\left\{ \ket{\tilde \phi_k} \right\}_k$. We can thus substitute $A \tilde G^{(H)} A^{-1} = G^{-1} G^{(H)}$, such that the Ritz values $\mu_k$ are indeed the eigenvalues of the Rayleigh matrix.
\end{proof}

The results we have given so far on the Ritz values can be illustrated by the special case where $V$ is an invariant subspace of $H$. Then, the restriction of $H$ to $V$ is well-defined, and it can be shown to be the operator $G^{-1} G^{(H)}$. Therefore, in this special case, the Ritz values are actual eigenvalues of $H$.

Since each Ritz value $\mu$ is an eigenvalue of the Rayleigh matrix, it has an associated eigenvector $\alpha$. From this eigenvector, we can define a state $\ket{\psi_\alpha} = \sum_p \alpha_p \ket{\phi_p}$ of the full Hilbert space, called \textit{Ritz vector}. In the next proposition, we show that the Ritz vector $\ket{\psi_\alpha}$ has energy $\mu$.

\begin{proposition}[Energy of Ritz vectors]
The eigenvectors $\alpha^{(k)}$ of $G^{-1} G^{(H)}$ corresponding to the eigenvalues $\mu_k$ define states of the full Hilbert space $\ket{\psi_{\alpha^{(k)}}} = \sum_p \alpha_p^{(k)} \ket{\phi_p}$ such that the average value of $\ket{\psi_{\alpha^{(k)}}}$ on $H$ is $\mu_k$.
\end{proposition}

\begin{proof}
Consider $\alpha$, an eigenvector of $G^{-1} G^{(H)}$ corresponding to the eigenvalue $\mu$. Then, the energy of the associated state $\ket{\psi_\alpha} = \sum_p \alpha_p \ket{\phi_p}$ is given by
\begin{align}
\frac{\braket{\psi_\alpha | H | \psi_\alpha}}{\braket{\psi_\alpha | \psi_\alpha}} &= \frac{\sum_{p,p'} \alpha_p^* \braket{\phi_p | H | \phi_{p'}} \alpha_{p'}}{\sum_{p,p'} \alpha_p^* \braket{\phi_p | \phi_{p'}} \alpha_{p'}} \\
&= \frac{\alpha^\dag G^{(H)} \alpha}{\alpha^\dag G \alpha} \\
&= \mu \frac{\alpha^\dag G \alpha}{\alpha^\dag G \alpha} \\
&= \mu,
\end{align}
where we used $G^{(H)} \alpha = \mu G \alpha$.
\end{proof}

This result means that the Ritz values $\mu_k$ are not just abstract quantities. They are associated with states of the full Hilbert space that have energy $\mu_k$, and that can be easily constructed.

We have thus seen how the energy of the Ritz vectors $\ket{\psi_{\alpha^{(k)}}}$ can be trivially obtained. We will now see how the average value of other observables can be obtained. While there exist other methods of computing observables of the Ritz vectors as detailed in \cref{sec:estimate_observables}, the following method has the advantage of only involving Rayleigh matrices.

\begin{proposition}[Average values of Ritz vectors]
\label{prop:ritz_vector_average}
Let $A$ be an operator. If the Ritz values are all distinct, then the average value of the $k$-th Ritz vector $\ket{\psi_{\alpha^{(k)}}}$ is given by
\begin{equation}
\frac{\braket{\psi_{\alpha^{(k)}} | A | \psi_{\alpha^{(k)}}}}{\braket{\psi_{\alpha^{(k)}} | \psi_{\alpha^{(k)}}}} = \left[ P^{-1} G^{-1} G^{(A)} P \right]_{kk},
\end{equation}
where $P$ is the matrix of eigenvectors of $G^{-1} G^{(H)}$.
\end{proposition}

\begin{proof}
We first need to prove that under the assumption that the Ritz values are all distinct, the Ritz vectors are orthogonal.

The eigendecomposition of the Rayleigh matrix $G^{-1} G^{(H)} = P D P^{-1}$ defines a family of Ritz vectors whose matrix of columns is $U = \Phi P$, where $\Phi_{sk} = \braket{s | \phi_k}$. If we define $Q = \Phi G^{-1} \Phi^\dag$ the orthogonal projector on the subspace spanned by the states $\ket{\phi_k}$ and $H_Q = Q H Q$ the Hamiltonian projected in this subspace, we notice that
\begin{align}
H_Q U &= \Phi G^{-1} \Phi^\dag H \Phi G^{-1} \Phi^\dag \Phi P \\
&= \Phi G^{-1} G^{(H)} P \\
&= \Phi P D \\
&= U D,
\end{align}
such that the Ritz vectors are also eigenvectors of the projected Hamiltonian $H_Q$, with the Ritz values as the corresponding eigenvalues.

Under the assumption that the Ritz values are all different, the Ritz vectors are thus eigenvectors of a Hermitian operator for distinct eigenvalues and are therefore orthogonal (but not necessarily orthonormal).

\paragraph{}
We denote by $\ket{\psi_{\alpha^{(k)}}}$ the Ritz vectors, that is, the columns of the matrix $U$. Let $\left\{ \ket{\psi_p} \right\}_p$ be an orthogonal family of $n-m$ vectors orthogonal to all states $\ket{\psi_{\alpha^{(k)}}}$, that completes $\left\{ \ket{\psi_{\alpha^{(k)}}} \right\}_k$ into an orthogonal basis $\mathcal B$. We denote by $\Psi$ the matrix of columns of the family $\left\{ \ket{\psi_p} \right\}_p$.

Let $R$ be the matrix of the basis $\mathcal B$ in the canonical basis, that is,
\begin{equation}
R = \begin{pmatrix}
\Phi P & \Psi
\end{pmatrix}.
\end{equation}
Then, the matrix of $A$ in the basis $\mathcal B$ can be expressed as\footnote{We make the usual abuse of notation of denoting both the operator and its matrix in the canonical basis by the same symbol $A$.}
\begin{align}
A_{\mathcal B} &= R^{-1} A R \\
&= R^{-1} \left( R^\dag \right)^{-1} R^\dag A R \\
&= \left( R^\dag R \right)^{-1} R^\dag A R.
\end{align}

Let us examine the block structure of $R^\dag R$,
\begin{align}
R^\dag R &=
\begin{pmatrix}
P^\dag \Phi^\dag \\
\Psi^\dag
\end{pmatrix}
\begin{pmatrix}
\Phi P & \Psi
\end{pmatrix} \\
&=
\begin{pmatrix}
P^\dag \Phi^\dag \Phi P & P^\dag \Phi^\dag \Psi \\
\Psi^\dag \Phi P & \Psi^\dag \Psi
\end{pmatrix} \\
&=
\begin{pmatrix}
P^\dag G_\phi P & 0 \\
0 & G_\psi
\end{pmatrix},
\end{align}
where the off-diagonal blocks are zero since the states $\ket{\psi_{\alpha^{(k)}}}$ and $\ket{\psi_p}$ are orthogonal. We use the indices of $G$ to distinguish between the Gram matrices of the states $\ket{\phi_k}$ and $\ket{\psi_p}$.

Let us then write down the blocks of $R^\dag A R$,
\begin{equation}
R^\dag A R =
\begin{pmatrix}
P^\dag G_\phi^{(A)} P & P^\dag G_{\phi \psi}^{(A)} \\
G_{\psi \phi}^{(A)} P & G_\psi^{(A)}
\end{pmatrix}.
\end{equation}

Combining both expressions, we finally get
\begin{equation}
A_{\mathcal B} =
\begin{pmatrix}
P^{-1} G_\phi^{-1} G_\phi^{(A)} P & P^{-1} G_\phi^{-1} G_{\phi \psi}^{(A)} \\
G_\psi^{-1} G_{\psi \phi}^{(A)} P & G_\psi^{-1} G_\psi^{(A)}
\end{pmatrix}.
\end{equation}

Since $A_{\mathcal B}$ is the matrix of $A$ in an orthogonal basis, we know that for $i,j \leq m$,
\begin{equation}
\left[ P^{-1} G_\phi^{-1} G_\phi^{(A)} P \right]_{ij} = \frac{\braket{ \psi_{\alpha^{(i)}} | A | \psi_{\alpha^{(j)}}}}{\braket{\psi_{\alpha^{(i)}} | \psi_{\alpha^{(i)}}}}.
\end{equation}
In particular, average values on a Ritz vector can be obtained as
\begin{equation}
\left[ P^{-1} G_\phi^{-1} G_\phi^{(A)} P \right]_{ii} = \frac{\braket{\psi_{\alpha^{(i)}} | A | \psi_{\alpha^{(i)}}}}{\braket{\psi_{\alpha^{(i)}} | \psi_{\alpha^{(i)}}}}.
\end{equation}
\end{proof}

\section{Variational principle for linear variational states}
\label{app:variational_principle}

We derive the (normalized) variational principle for a linear variational state, and show that it is in fact equivalent to the unnormalized variational principal in this case. This form is more simple and convenient than the direct form of the normalized variational principle.

\begin{proposition}
Let $\left\{ \ket{\phi_k} \right\}_k$ be a family of linearly independent states, and let $\ket{\psi_\alpha}$ be the linear variational state defined on this basis,
\begin{equation}
\ket{\psi_\alpha} = \sum_k \alpha_k \ket{\phi_k}, \ \alpha \in \mathbb C^m.
\end{equation}
Then, the time-dependent variational principle applied to $\ket{\psi_\alpha}$ for a Hamiltonian $H$ results in the variational dynamics
\begin{equation}
\dot \alpha = - i G^{-1} G^{(H)} \alpha.
\end{equation}
\end{proposition}

\begin{proof}
We recall that for a given parameter $\alpha$, the time-dependent variational principle gives the variation of the parameter as a solution of the equation
\begin{equation}
\label{eq:tVMC}
S \dot \alpha = -i F,
\end{equation}
where the quantum geometric tensor $S$ and forces vector $F$ are given in general by
\begin{align}
S_{pp'} &= \frac{\partialps \bra{\psi_\alpha} \partialpp \ket{\psi_\alpha}}{\braket{\psi_\alpha | \psi_\alpha}} - \frac{( \partialps \bra{\psi_\alpha} ) \ket{\psi_\alpha}}{\braket{\psi_\alpha | \psi_\alpha}} \frac{\bra{\psi_\alpha} \partialpp \ket{\psi_\alpha}}{\braket{\psi_\alpha | \psi_\alpha}}, \\
F_p &= \frac{( \partialps \bra{\psi_\alpha} ) H \ket{\psi_\alpha}}{\braket{\psi_\alpha | \psi_\alpha}} - \frac{\braket{\psi_\alpha | H | \psi_\alpha}}{\braket{\psi_\alpha | \psi_\alpha}} \frac{( \partialps \bra{\psi_\alpha} ) \ket{\psi_\alpha}}{\braket{\psi_\alpha | \psi_\alpha}}.
\end{align}

In the case of the linear variational state $\ket{\psi_\alpha}$, the derivatives are easily obtained as
\begin{equation}
\partialp \ket{\psi_\alpha} = \ket{\phi_p}.
\end{equation}
Plugging these derivatives into $S$ and $F$, we get
\begin{align}
S &= \frac{1}{\alpha^\dag G \alpha} \left[ G - \frac{G \alpha \alpha^\dag G}{\alpha^\dag G \alpha} \right], \\
F &= \frac{1}{\alpha^\dag G \alpha} \left[ G^{(H)} \alpha - \frac{\alpha^\dag G^{(H)} \alpha}{\alpha^\dag G \alpha} G \alpha \right],
\end{align}
where $G_{ij} = \braket{\phi_i | \phi_j}$  and $G_{ij}^{(H)} = \braket{\phi_i | H | \phi_j}$ are respectively the Gram matrix and projected Hamiltonian corresponding to the family $\left\{ \ket{\phi_k} \right\}_k$.

We notice that we always have $S \alpha = 0$, such that $S$ is not invertible. Thus, we cannot simply take $\dot \alpha = - i S^{-1} F$ as our variation of the parameter. The linear system has either no solution or has a solution space which is an affine space of dimension larger than 1. It is easy to see that one particular solution is $- i G^{-1} G^{(H)} \alpha$, and that the set of all solutions is thus given by
\begin{equation}
- i G^{-1} G^{(H)} \alpha + \ker S.
\end{equation}

We now need to identify the kernel of $S$. Since $G$ has full rank $m$ and $G \alpha \alpha^\dag G = G \alpha (G \alpha)^\dag$ is a rank 1 matrix, then $S$ can only have rank $m$ or $m-1$. But we know that $S$ is singular, so it must have rank $m-1$. As such, we have $\ker S = \operatorname{span}(\alpha)$ and the set of solutions to \cref{eq:tVMC} is
\begin{equation}
\left\{- i G^{-1} G^{(H)} \alpha + \lambda \alpha, \ \lambda \in \mathbb C \right\}
\end{equation}
for any $\lambda \in \mathbb C$.

Let us assume that we pick an arbitrary solution at all times, in such a way that $\lambda(t)$ is an integrable function with antiderivative $\Lambda(t)$. Then the resulting dynamics is given by
\begin{equation}
\label{eq:big_lambda}
\alpha(t) = e^{-i \left( G^{-1} G^{(H)} + \Lambda(t) \right)} \alpha_0.
\end{equation}
Since $\Lambda(t)$ is a scalar function, it doesn't change the resulting quantum state. The function
\begin{equation}
\alpha(t) = e^{-i G^{-1} G^{(H)}} \alpha_0
\end{equation}
describes the same quantum state as \cref{eq:big_lambda}, meaning that the choice of $\lambda$ is irrelevant. We can thus always choose the solution with $\lambda=0$. We conclude that in this case, the TDVP equation is equivalent to
\begin{equation}
\dot \alpha = - i G^{-1} G^{(H)} \alpha.
\end{equation}
\end{proof}

\section{Estimators of the Rayleigh matrix}
\label{app:rm_estimators}

In this section we explore in more details the various possible estimators of the Rayleigh matrix $G^{-1} G^{(H)}$. The two classes of estimators that we consider are the determinant state estimator from \cref{eq:rm_decomposition} and the sum of states estimator from \cref{eq:sum_of_states_1,eq:sum_of_states_2}.
One remark we can make concerning the sum of states estimator is that as the number of states $m$ increases, the probability distribution it relies on becomes increasingly unadapted to the specific matrix elements it estimates. This might lead to an increase in the variance of this estimator. Compared to the determinant state estimator, however, it does have the advantage of being much cheaper to compute for a given number of samples (see \cref{tab:estimator_costs}).
We also note that both estimators reduce to the standard estimator of the energy when $m=1$.

We now discuss various details about these two estimators that lead to different possible implementations, each of which we benchmark in \cref{fig:estimator_difference_big}.

\paragraph{Keep or discard imaginary parts}
When estimating the energy of a single state using the local energy, the resulting estimate is complex and only becomes real in the limit of infinite samples. As such, the real part of the resulting estimate is taken. Similarly, the Rayleigh matrix is known to have real eigenvalues, but the estimators using a finite number of samples will produce estimates that have eigenvalues with a non-zero imaginary part. As such, we can transform the resulting estimate by discarding the imaginary part of its eigenvalues. In \cref{fig:estimator_difference_big}, for all estimators, we examine both keeping and discarding the imaginary part of the eigenvalue in order to identify which variant is more efficient.

\paragraph{Numerical stability}
The sum of states estimator produces $G$ and $G^{(H)}$ individually, such that they still need to be combined into $G^{-1} G^{(H)}$. The inversion of $G$ can be especially ill-conditioned as the states $\ket{\phi_k}$ get close to being linearly dependent. This means that the Monte Carlo error incurred by estimating $G$ stochastically is significantly amplified when taking its inverse. We therefore examine two methods to construct $G^{-1} G^{(H)}$. The first method consists in taking the inverse in high precision. The second method consists in taking the pseudo-inverse for a given singular value cut-off. We compare those two methods in \cref{fig:estimator_difference_big}.

We note that the determinant state estimator suffers much less from the ill-conditionedness of $G$ than the sum of states estimator. In the sum of states estimator, the matrix $G$ that we invert is known up to Monte Carlo error $\epsilon_{\text{MC}}$. If we denote the condition number of $G$ by $\kappa$, the resulting error on $G^{-1}$ is bounded by $\kappa \epsilon_{\text{MC}}$. However, in the case of the determinant state estimator, we instead invert $\Phi(\bm s)$, which is known up to numerical precision $\epsilon_{\text{NP}} \sim 10^{-15}$. Assuming that the condition numbers of the matrices $\Phi(\bm s)$ for the sampled configurations $\bm s$ are of similar magnitude to $\kappa$, the resulting error is instead bounded by $\kappa \epsilon_{\text{NP}}$. Since typically $\epsilon_{\text{NP}} \ll \epsilon_{\text{MC}}$, we see that the determinant state estimator is expected to be significantly more stable to large $\kappa$ than the sum of states estimator.

\begin{figure}[tb]
\centering
\includegraphics[width=0.9\linewidth]{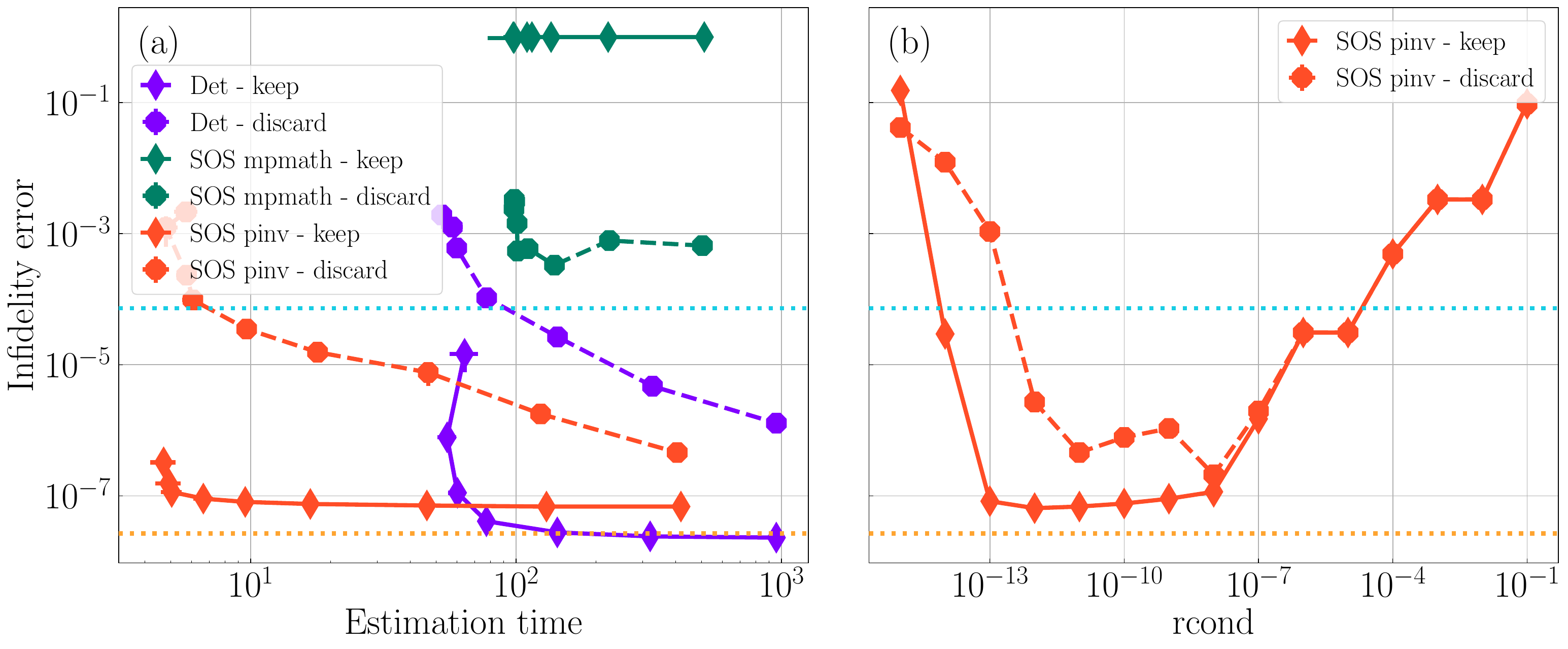}
\caption{(a) Comparison between the final infidelity error of Bridge performed using various estimators of the Rayleigh matrix as a function of the estimation time of the Rayleigh matrix. The Hamiltonian is the transverse-field Ising Hamiltonian on a $4 \times 4$ lattice with $(h,J) = (1,0.1)$, and the basis states are 136 CNN states obtained with p-tVMC from time 0 to 1.35 with $\delta=0.01$. The Rayleigh matrix is estimated either with the determinant state (Det) estimator of \cref{eq:rm_decomposition} or with the sum of states (SOS) estimator of \cref{eq:sum_of_states_1,eq:sum_of_states_2} for various numbers of samples. For the sum of states, the inverse of $G$ was either taken exactly with 400 digits of precision (mpmath) or using the pseudo-inverse (pinv) with an singular value cut-off of $10^{-11}$. For each of the three methods, full lines show the performance when keeping the imaginary part of the eigenvalues of the Rayleigh matrix, and dashed lines show the performance when discarding them. The dotted horizontal orange line indicated the optimal performance that can be achieved in the subspace. In theory, it should be below all other points, but it is subject to numerical error. The dotted horizontal blue line indicates the performance of the original states $\ket{\phi_k}$. (b) Performance of Bridge when estimating the Rayleigh matrix with the sum of states estimator and inverting $G$ with the pseudo-inverse for various choices of the singular value cut-off rcond. The estimation is performed with $10^5$ samples which corresponds to the rightmost point of the corresponding curve in plot (a).}
\label{fig:estimator_difference_big}
\end{figure}

\paragraph{Numerical comparison}
In order to compare the performance of the two estimators and their variants, we could compare the resulting estimates of the Rayleigh matrix to the exact Rayleigh matrix. However, when the states $\ket{\phi_k}$ start being almost linearly dependent, the exact calculation of the Rayleigh matrix starts being numerically unstable. While the computation of $G^{-1} G^{(H)}$ from $G$ and $G^{(H)}$ can be performed in high precision, the exact computation of $G$ and $G^{(H)}$ alone must be performed in double precision, since the overhead associated to using unoptimized high-precision libraries like mpmath makes it unpractical to use in this case. The small numerical error accumulated in the computation of $G$ and $G^{(H)}$ is then amplified when taking the ill-conditioned inverse of $G$. 

Since the exact calculation of the Rayleigh matrix cannot be trusted, we instead use the performance of Bridge as a benchmark, where the exact dynamics used as the reference can be computed without significant numerical error. In figure \ref{fig:estimator_difference_big}, we show the performance of Bridge using the determinant state estimator, the sum of states estimator with high precision inversion of $G$ and the sum of states estimator with inversion of $G$ using the pseudo-inverse at various singular value cut-offs. For each, we either keep or discard the imaginary part of the eigenvalues of the resulting estimate of the Rayleigh matrix. For more details on the setting of the dynamics, see \cref{sec:numerical_results}. We conclude from this figure that the optimal estimator is either the sum of states estimator that keeps the imaginary part with an appropriately chosen singular value cut-off for a good performance and a very fast estimation (less than 10 seconds), or the determinant state estimator that keeps the imaginary part for the optimal performance at an increased cost (about 5 minutes).

In the context of Bridge, irrespective of the chosen method, the calculation time will be dominated by the generation of the original basis states $\ket{\phi_k}$. Therefore, we use the determinant state estimator that keeps the imaginary part of the eigenvalues, since it achieves the optimal final performance and does not require tuning the singular value cut-off.

We further show that the determinant state estimator is more suitable for Bridge than the pseudo-inverse sum of states estimator with \cref{fig:8x8_with_sos} which is \cref{fig:8x8} with an additional curve depicting Bridge performed with the pseudo-inverse sum of states estimator for a given value of the pseudo-inverse singular value cut-off. It confirms that the determinant state estimator also outperforms the sum of states estimator for larger system sizes. In particular, Bridge with the sum of states estimator is even less accurate than the original basis states themselves in the case (a). This is further aggravated by the fact that in both cases, the sum of states estimator ran for significantly longer. In all cases, Bridge was run on a single A100 GPU. With the determinant state estimator, it took 15 minutes for (a) and 60 minutes for (b). With the sum of states estimator, it took 150 minutes for both (a) and (b).

\begin{figure}[t]
\centering
\includegraphics[width=\linewidth]{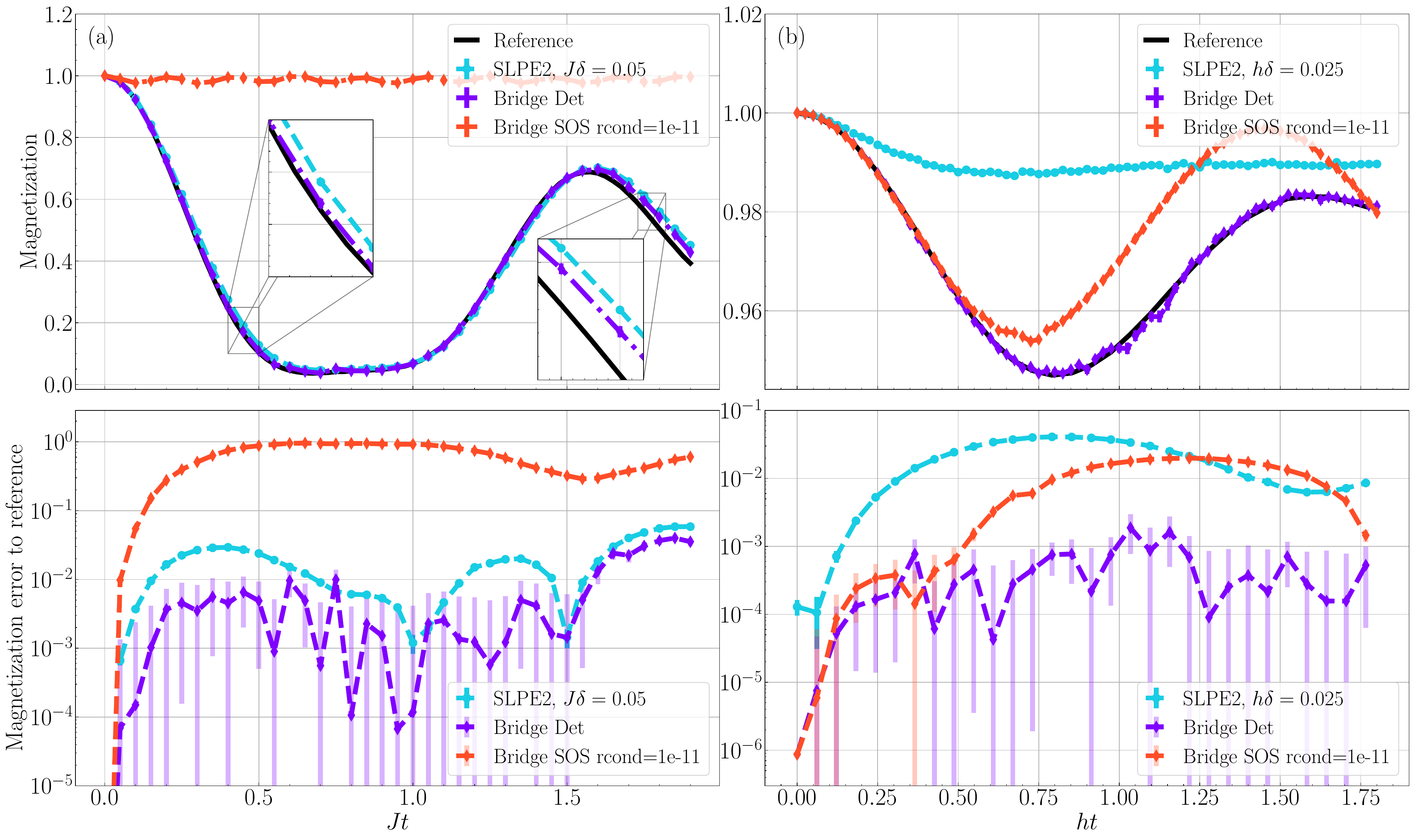}
\caption{Magnetization on the $8 \times 8$ transverse-field Ising model with $h=0.1h_c$ (a) and $h=2h_c$ (b). The original states are vision transformers optimized with p-tVMC using the scheme SLPE2 with time step $J\delta = 0.05$ (a) and $h\delta = 0.025$ (b). The reference for (a) are vision transformers optimized with p-tVMC using the higher order scheme SLPE3 with a smaller time step $J\delta = 0.025$. The reference for (b) is convolutional neural network states obtained with t-VMC in Ref.~\cite{Schmitt2020MarkusMarkus}, since obtaining good p-tVMC states in this regime would have been too computationally expensive. In both cases, Bridge is performed with either the determinant state estimator (Det) or sum of states estimator (SOS). For the sum of states estimator, $G$ is inverted with the pseudo-inverse, with a singular value cut-off $\texttt{rcond}=10^{-11}$.}
\label{fig:8x8_with_sos}
\end{figure}

\section{Projected-time variational Monte Carlo}
\label{app:ptvmc}
In this section, we summarize the core features of the projected-time variational Monte Carlo (p-tVMC) method, referring the reader to \cite{Sinibaldi2023Unbiasing, Gravina2024PTVMC} for a more comprehensive treatment. 
The main goal of p-tVMC is to numerically solve the time-dependent Schrödinger equation $\ket{\psi(t+\delta)} = e^{\Lambda \delta} \ket{\psi(t)}$
on a discrete time mesh with step size $\delta$.
We assume that the exact dynamics is well approximated by a variational state, meaning that there exist parameters $\theta(t)$ such that $\ket{\psi(t)} \approx \ket{\psi_{\theta(t)}}$. Using the McLachlan variational principle, the evolution problem is mapped onto an optimization problem in parameter space: \begin{equation} 
\label{eq:minimization_problem}
\theta(t+\delta) = \argmax{\eta}\,\,\mathcal{F}\left(\ket{\psi_\eta}, e^{-iH \delta}\ket{\psi_{\theta(t)}}\right), 
\end{equation} 
with the fidelity
\begin{equation} 
\mathcal{F}(\ket{\psi}, \ket{\phi}) = \frac{\braket{\psi|\phi}\braket{\phi|\psi}}{\braket{\psi|\psi}\braket{\phi|\phi}}
= \mathbb{E}_{\raisebox{1em}{$\substack{x\sim|\psi(x)|^2 \ y\sim|\phi(y)|^2}$}}\left[\frac{\phi(x)}{\psi(x)}\frac{\psi(y)}{\phi(y)}\right]
\end{equation}
measuring the discrepancy between two quantum states.

The p-tVMC approach leverages an expansion of the exponential map $e^{\Lambda\delta}$ to finite order in $\delta$, yielding discretizations specifically suited for the optimization in \cref{eq:minimization_problem} \cite{Sinibaldi2023Unbiasing,Gravina2024PTVMC,Nys2024,Medvidovi2021,Donatella2023}. 
These expansions typically assume product forms, such as 
\begin{equation} \label{eq:LPE} e^{\Lambda \delta} = \prod_{k=1}^s \left( \mathbb 1 - i a_k H \delta \right) + O \left( \delta^{s + 1} \right) , \quad \text{or} \quad e^{\Lambda \delta} = \prod_{k=1}^s \left( \mathbb 1 - i a_k H_1 \delta \right) e^{-i b_k H_0 \delta} + O\left(\delta^{o(s) + 1}\right), 
\end{equation} 
where $H_0$ is diagonal in the computational basis, and $H_1$ captures the off-diagonal part, with $H = H_0 + H_1$. 
Importantly, the number of substeps $s$ scales polynomially in the order $o=o(s)$ of the approximation.
These discretizations are known as the \emph{Linear Product Expansion} (LPE) and the \emph{Splitted Linear Product Expansion} (SLPE) respectively.
The values of the coefficients ${a_k, b_k}$ are determined semi-analytically and are tabulated in \cite{Gravina2024PTVMC}.
A key advantage of these expansions is their ability to achieve high-order temporal accuracy by approximately decomposing the single optimization step [\cref{eq:minimization_problem}] into a sequence of $s$ smaller and efficiently computable optimization problems. Introducing the intermediate parameters ${\theta^{(k)}}$, with initial condition $\theta^{(0)} \equiv \theta(t)$ and final result $\theta^{(s)} \equiv \theta(t+\delta)$, the procedure becomes: 
\begin{equation}
\label{eq:optimization_UV} 
\theta^{(k)} = \argmax{\eta}\,\,\mathcal{F} \left( \ket{\psi_\eta}, U_k \ket{\psi_{\theta^{(k-1)}}} \right), \quad k \in \left\{ 1, \dots, s \right\}, 
\end{equation} 
where each operator $U_k$ corresponds to a factor in the chosen product expansion \cite{Gravina2024PTVMC, Nys2024}.

The error associated with the p-tVMC method arises from two distinct sources. Firstly, a finite \emph{discretization error} remains due to truncation of the expansions of the exponential map [\cref{eq:LPE}]. It can be systematically controlled and made arbitrarily small as the time step $\delta \to 0$, at the expense of an increased computational cost. Secondly, each optimization step introduces an \emph{optimization error}, arising from finite sampling errors, limited expressivity of the variational ansatz, and challenging optimization landscapes. This error accumulates over successive steps and can significantly affect the overall accuracy. 
The optimization problems at each substep are typically solved using natural gradient descent combined with Monte Carlo approximations of the loss function and its gradient, as detailed in \cite{Sinibaldi2023Unbiasing,Gravina2024PTVMC}.

\end{document}